\documentclass[11pt]{article}
\usepackage[letterpaper, margin=1in]{geometry}
\usepackage{amsmath, amssymb, amsthm}
\usepackage[hybrid,hashEnumerators,smartEllipses]{markdown}
\usepackage{bbm}
\usepackage{graphicx}
\usepackage{url}
\usepackage{hyperref}
\usepackage[capitalise]{cleveref}
\usepackage{algorithm}
\usepackage[noend]{algpseudocode}
\usepackage{subcaption}
\usepackage{mathrsfs}
\usepackage{xcolor}
\usepackage{xspace}
\usepackage{tikz}
\usetikzlibrary{decorations.pathreplacing}
\usetikzlibrary{positioning}

\newtheorem{theorem}{Theorem}[section]
\newtheorem{corollary}[theorem]{Corollary}
\newtheorem{lemma}[theorem]{Lemma}

\newtheorem{definition}[theorem]{Definition}
\newtheorem{claim}[theorem]{Claim}

\newtheorem{example}[theorem]{Example}

\crefname{claim}{Claim}{Claims}
\Crefname{claim}{Claim}{Claims}

\usepackage{fancyhdr}

\newcommand{\remove}[1]{}
\newcommand{\ignore}[1]{}
\newcommand{\mc}{\mathcal}

\newcommand{\cost}{\mathrm{cost}}

\newcommand{\OPT}{\mathrm{OPT}}
\newcommand{\Match}{\mathrm{Match}}

\def\R{\mathbb{R}}

\newcommand{\E}{\mathbb E}

\usepackage{enumitem}
\hypersetup{
    colorlinks=true,
    linkcolor=black,
    filecolor=blue,
    urlcolor=blue,
    citecolor=blue,
    pdfpagemode=FullScreen,
}

\newcommand{\rad}{\mathrm {rad}}
\usepackage[natbib=true, url=false, style=alphabetic, doi=false, isbn=false, maxbibnames=9,maxalphanames=4, minalphanames=3]{biblatex}

\AtEveryBibitem{%
  \clearfield{url}%
  \clearfield{urldate}%
  \clearname{editor}%
  \clearfield{eprint}
}
\AtEveryBibitem{%
  \clearlist{language}
  \clearfield{location}
}
\newcommand{\mov}{\mathrm{move}}
\newcommand{\clust}{\mathrm{cluster}}
\newcommand{\Clust}{\mathrm{Cluster}}
\renewcommand{\tau}{u}
\renewcommand{\d}{\mathrm{d}}
\title{Competitive Online Clustering with Movement Cost}
\author{Christian Coester\footnote{Funded by the European Union (ERC, CCOO, 101165139).
Views and opinions expressed are however those of the author(s) only and do not necessarily reflect
those of the European Union or the European Research Council. Neither the European Union nor
the granting authority can be held responsible for them. Email:~\texttt{christian.coester@cs.ox.ac.uk}.}\\University of Oxford \and Yichen Huang\footnote{Supported in part by NSF grant CNS-2107078. Email: \texttt{yichenhuang@g.harvard.edu}.}\\Harvard University\and Michael Mitzenmacher\footnote{Supported in part by NSF grant CNS-2107078. Email: \texttt{michaelm@eecs.harvard.edu}.}\\Harvard University}
\begin{document}
\maketitle
\pagenumbering{gobble}

\begin{abstract}
We study the problem of online clustering with movement cost, which generalizes the fundamental $k$-server problem. At each time step, the algorithm is given a set of points in a metric space and must choose $k$ cluster centers. The cost incurred is the sum of the clustering cost at each time step and the total distance that the centers move over time. Prior work established competitive algorithms for the case $k\in\{1,2\}$ on the line metric with the $k$-median clustering objective;   under the $k$-median objective, this problem is also known as facility reallocation. We substantially generalize these results to arbitrary $k$, general metric spaces, and a broad class of clustering objectives.

Specifically, we consider clustering costs defined as a norm of the distance vector from points to their nearest center, which subsumes the $k$-median, $k$-center, and (square root of) $k$-means objectives, and provide an asymptotically optimal $O(k)$-competitive randomized algorithm against adaptive online adversaries. The algorithm is based on a reduction to the $k$-server problem, establishing a tight connection between the problems. In this reduction, $k$-server requests are generated by Poisson processes, whose arrival rates evolve in continuous time and are driven by a dual representation of the clustering costs. Our method of using Poissonization in this context may also be of independent interest. We demonstrate another use of the idea by deriving similar reductions for batched set chasing problems.

We also consider a different family of clustering objectives defined as a norm of the vector of cluster radii, which subsumes the sum-$k$-radii and the $k$-center objectives, and give a tight $O(k)$-competitive deterministic algorithm on the line.
\end{abstract}

\newpage

\pagenumbering{arabic}
\section{Introduction}
\label{sec: intro}

Clustering is a fundamental question in algorithms and machine learning, with broad applications in data summarization, facility location, resource allocation, and beyond~\cite{CharikarG99,lin92approximation,AryaGKMMP01,CharikarGuhaTardosShmoys99,Young00,JainVazirani01,KumarSS10,Kumar12,CharikarLi12,LiSvensson16,AlamdariShmoys17}. In a clustering problem, the algorithm is given a set of points (``requests'') and must select $k$ centers so that the requests are well served by the centers according to some objective (e.g., minimizing the total distance from requests to their nearest center).

In many modern settings, however, requests are not static: they arrive and depart over time, and their locations may evolve~\cite{EisenstatMS14,FotakisKZ21,TantipathananandhBK07,AnNS17}. Thus, an algorithm might choose to update the locations of the centers, leading to extra costs proportional to the distance moved~\cite{FOTAKIS202113, Christou23Efficient, CompetitiveConsistent, LackiHGJR24, ForsterS25}. For example, mobile disease-testing units may relocate as new outbreak hotspots emerge, food trucks may change locations to adapt to expected demands (e.g. events, lunch hours, nightlife), and edge-computing resources may be repositioned and migrated as user demand shifts across regions~\cite{karanik, liang2023}. This forces a trade-off between the \emph{quality} of the clustering at each time and the \emph{stability} of the chosen centers over time, and more challengingly, the algorithm must seek the balance without knowing the future requests.

Reasoning about movement without knowledge of future requests is a central theme in online algorithms. A prominent relevant example is the $k$-server problem~\cite{MMS88}: $k$ servers are located at points in a metric space $(X,d)$. At each time $t$, a request appears at a point $x_t \in X$, and the algorithm must move a server to $x_t$. The cost is the total distance traveled by the servers. The performance of an online algorithm is measured by its \emph{competitive ratio}, defined as the worst-case ratio between the cost of the algorithm and that of an omniscient optimum who knows the sequence in advance. The $k$-server problem has long been one of the most important problems in online algorithms and has led to many deep insights~\cite{CKP91new,KP95,Bubeck18,BubeckCR23,CoesterKZ26}.

In this paper, we study the following simultaneous generalization of the clustering problem and the $k$-server problem. We give the most general definition here, and later specify the assumptions on the cost functions used in our results.

\begin{definition}[Online Clustering with Movement Cost (CMC)]
\label{def:problem}
    Let $(X, d)$ be a metric space with $k$ centers initially located at a set $C_0$. At each round $t$, the algorithm observes a request multiset $R_t\subseteq X$ and chooses a set $C_t\subseteq X$ of $k$ centers. The cost incurred in round $t$ is
\[
    c_t \;=\; \underbrace{\Clust_t(R_t, C_t)}_{\text{clustering}} \;+\;
    \underbrace{\Match(C_{t-1}, C_t)}_{\text{movement}},
\]
where $\Match(C_{t-1}, C_t)$ denotes the minimum matching distance between the configurations $C_{t-1}$ and $C_t$ under the metric $d$. The clustering cost $\Clust_t$ is flexible so as to capture different objectives.
\end{definition}

The CMC framework captures $k$-server as a special case, since a $k$-server request at $x_t$ can be simulated by repeated requests to the singleton set $\{x_t\}$, thereby forcing any competitive algorithm to move a center\footnote{A server in the $k$-server instance corresponds to a center in a CMC instance.} to $x_t$. Conversely, \cite{bubeck21online} showed that CMC with the $k$-median objective and \emph{singleton} request sets can be reduced to $k$-server as well. They further highlighted that the problem can be viewed in the line of chasing (well-behaved) convex functions, but in a multi-server rather than the usual single-server setting.

When request sets are not singletons, the situation becomes more complex. The centers must minimize the distance to multiple requests simultaneously, possibly more than the number of centers when $|R_t|>k$. For this more general setting, prior work established competitive online algorithms only for the case $k \le 2$ when the metric space is the line and the clustering cost is the $k$-median objective~\cite{FOTAKIS202113,de_keijzer_facility_2022}, referred to as facility reallocation on the line. We substantially generalize these results with respect to $k$, the metric space, and the clustering objective. Specifically, we consider the following two fairly broad families of clustering costs.

\paragraph{Norm-of-distances.} Perhaps the most natural clustering cost is a norm of the vector of distances from each request to its nearest center. This subsumes the $k$-median objective (sum of distances) via the $\ell_1$-norm, the $k$-center objective (maximum distance) via the $\ell_\infty$-norm, and the square root of the $k$-means objective via the $\ell_2$-norm. In the clustering context, the norm should naturally be \emph{monotone}, i.e., if $0 \le u_i \le v_i$ for all coordinates $i$, then $\|u\| \le \|v\|$, and we will henceforth assume monotonicity. We also show in \cref{thm:non-monotone-lower} that no competitive algorithm exists without the monotonicity assumption, even when $k=1$.

\paragraph{Norm-of-radii.} An alternative notion is that the costs are associated with the clusters as a whole instead of individual requests, and we consider the clustering cost to be a norm of the vector of cluster radii. This subsumes the sum-$k$-radii objective~\cite{Fotakis2011OnlineSC,charikar2001cluster} via the $\ell_1$-norm and the $k$-center objective via the $\ell_\infty$-norm.\footnote{Note that for the $\ell_\infty$-norm, the norm-of-distances model and the norm-of-radii model are identical.} For this objective, in addition to monotonicity, we require the norm to be symmetric, i.e., invariant under permutation of centers.

Note that in the norm-of-radii model, the cost is determined not only by the center locations but also by the association of requests to centers. Our results for this model are specific to the line metric, and we show in Appendix~\ref{appendix: association} that on the line it is optimal up to a constant factor to associate each request with the closest center, so we may assume this association at the cost of a constant factor loss. For other metrics, one could alternatively study a model where the association is chosen by the algorithm; we discuss this matter (and how nearest association can be suboptimal in general metrics) in Appendix~\ref{appendix: association}.

\paragraph{Adversary Models.} The performance of an online algorithm is measured by the competitive ratio against an adversary who generates the requests and also serves them. We consider both oblivious and adaptive online adversaries~\cite{BenDavidBKTW94}. An oblivious adversary fixes the entire request sequence $\sigma$ in advance, before seeing the algorithm's randomness. An adaptive online adversary may choose each request based on the algorithm's past actions, and has to serve it immediately (and thus may not be optimal in hindsight). For deterministic algorithms, the two types of adversaries are equivalent.

\subsection{Our Results and Techniques}

We show that, for randomized algorithms against adaptive online adversaries, CMC with the norm-of-distances objective is no harder than the $k$-server problem.

\begin{theorem}[Norm-of-Distances, Randomized, General Metrics]
\label{thm:any-norm-intro}
There is a randomized $O(k)$-competitive algorithm for online CMC in the norm-of-distances model against adaptive online adversaries.
\end{theorem}

The bound is tight due to a lower bound of $k$ for the $k$-server problem~\cite{MMS88} (which also holds for randomized algorithms against adaptive online adversaries~\cite{Raghavan1994MemoryVR}). It is the first competitive algorithm for the problem for general $k$. This answers an open question from~\cite{FOTAKIS202113}, whose previous online algorithm only handled the case $k = 2$ on the line with the $k$-median objective.

By the relation between deterministic competitiveness and randomized competitiveness against adaptive online adversaries~\cite{BenDavidBKTW94}, \Cref{thm:any-norm-intro} also implies, non-constructively, the existence of an $O(k^2)$-competitive deterministic algorithm.

We obtain our result by a reduction to the $k$-server problem. A novelty of our analysis is a \emph{Poissonization} step employed in the reduction: we convert each per-round batch of requests in online CMC into a collection of independent Poisson arrival processes at the request locations, with instantaneous rates selected according to the current contributions of individual requests. A $k$-server algorithm then serves the request points generated by the Poisson processes, and our CMC algorithm follows a sequence of selected configurations of the $k$-server algorithm.

Poissonization (roughly, replacing a fixed number of events by a Poisson process) is, of course, a known technique for simplifying dependencies~\cite{JacquetSzpankowski1998AnalyticalDepoissonization,Kingman1993PoissonProcesses,MitzenmacherUpfal2017ProbabilityAndComputing}, and it has also been used, for example, in developing algorithms for stochastic online matching (e.g., \cite{HuangShu2021PoissonArrivalsNaturalLP,JailletLu2014OnlineStochasticMatching,qiu2023improved}). A key difference is that those applications typically concern \emph{intrinsically stochastic} settings (e.g., Poisson arrivals as in stochastic matching or secretary/prophet-type models), where randomness is part of the input. In contrast, our CMC problem is deterministic, and we use Poissonization to isolate each request's marginal contribution and aggregate these contributions independently into a randomized reduction. From this perspective, it plays a role somewhat reminiscent of rounding: the marginal contribution specifies fractional request intensities, and the Poisson process converts them into an integral request stream, in a spirit similar to~\cite{AnNS17}.

Moreover, unlike usual settings with fixed Poisson rates, the instantaneous rate assigned to each request $r$ evolves in continuous time and is \emph{state-dependent}, determined by its norming functional for the current clustering cost. This continuous-time Poissonization transforms each discrete batch into a time-dispersed stream that allows relating the clustering norm to the expected infinitesimal movement of the induced $k$-server process. We believe this idea may be of interest more broadly for ``linearizing'' batched objectives into continuous-time online processes.

As a concrete illustration, we generalize the reduction to set chasing problems~\cite{Chrobak1991TheSP,Friedman93convex,bubeck22shortest, sellke20chase,ArgueGGT20,Bubeck19competitive}. In $\mc F$-chasing, where $\mc F$ is a family of subsets of a metric space $(X,d)$, an algorithm receives a set $F_t\in\mc F$ at each time $t$; the algorithm must move to some $x_t\in F_t$ and pays $d(x_{t-1},x_t)$. We consider the batched variant of these problems, in which the algorithm receives a collection $R_t\subseteq\mc F$ each time. It outputs an arbitrary $x_t\in X$ and pays $d(x_{t-1},x_t)\;+\;\|(d(x_t,F))_{F\in R_t}\|$ for some norm $\|\cdot\|$ that is monotone.\footnote{Here $(d(x_t,F))_{F\in R_t}$ denotes the vector whose coordinate for each $F\in R_t$ is the distance from $x_t$ to the set $F$, namely $d(x_t,F)=\min_{y\in F} d(x_t,y)$. Note that a soft penalty is necessary, since it may happen that $\bigcap_{F\in R_t} F=\varnothing$. The classical setting with a single set $F_t$ and the hard constraint $x_t\in F_t$ can be recovered by repeating the request $\{F_t\}$ many times.} We show that the reduction underlying Theorem~\ref{thm:any-norm-intro} extends to any set chasing problem, and also to competitive algorithms against oblivious adversaries in case $\|\cdot\|$ is a weighted $\ell_1$-norm.

\begin{theorem}
\label{thm: mss}
If there is a randomized $\alpha$-competitive algorithm for $\mc F$-chasing against adaptive online adversaries, then there is a randomized $O(\alpha)$-competitive algorithm for the batched $\mc F$-chasing problem against adaptive online adversaries. If $\|\cdot\|$ is a weighted $\ell_1$-norm, then the reduction also holds for competitive ratios against oblivious adversaries.
\end{theorem}

Combined with the recently announced $O(\log^2 k)$-competitive algorithm for $k$-server against oblivious adversaries~\cite{OAI26}, this gives an $O(\log^2 k)$-competitive algorithm for CMC with the $k$-median objective against oblivious adversaries. It also implies an $O(d)$-competitive algorithm for batched convex body chasing on $\R^d$, an $O(2^w)$-competitive algorithm for batched size-$w$ set chasing against adaptive online adversaries, and if $\|\cdot\|$ is a weighted $\ell_1$-norm, an $O(w^2)$-competitive algorithm against oblivious adversaries.

The above reductions are inherently randomized because of the application of Poissonization. For the line metric, however, we also provide a \emph{deterministic} $O(k)$-competitive algorithm for CMC. The algorithm decomposes the contribution of requests, but instead of a Poissonization that issues requests probabilistically, we scale the movement speed due to each request according to its marginal contribution and sum the resulting velocities.

\begin{theorem}[Norm-of-Distances, Deterministic, Line]
\label{thm:aggrdc-intro}
On the line metric, there is a deterministic $O(k)$-competitive algorithm for online CMC in the norm-of-distances model.
\end{theorem}

Finally, we show that an adaptation of the algorithm of Theorem~\ref{thm:aggrdc-intro} also yields an analogous guarantee for the norm-of-radii objective.

\begin{theorem}[Norm-of-Radii, Deterministic, Line]
\label{thm:radii-intro}
    On the line metric, there is a deterministic $O(k)$-competitive algorithm for online CMC in the norm-of-radii model when the norm at every time is symmetric.
\end{theorem}

While the deterministic algorithm for norm-of-radii is similar to the algorithm for norm-of-distances, the analysis requires some new ideas. The challenge is that, unlike in the norm-of-distances setting, an offline center may simultaneously account for the farthest requests of several online clusters while paying only a single \emph{radius}. As a result, a naive charging argument from online centers to offline centers can overcount the offline cost and fail to yield a valid lower bound. We overcome this by exploiting the geometry of intervals on the line, together with the convexity and symmetry of the norm, to show that such overlaps can still be controlled within a constant factor.

\paragraph{Outline.} Following a discussion of related work in this section, \Cref{sec: prelim} defines the notation and provides some preliminaries. \Cref{sec: norm} presents the randomized algorithm for the norm-of-distances objective on general metrics and the generalization to chasing problems. \Cref{sec: line det} presents deterministic algorithms on the line. \Cref{sec: conclusion} contains some concluding remarks and open questions. Omitted proof details and further discussions can be found in the appendices.

\subsection{Related Work}
\label{sec: related work}
\paragraph{Online Clustering.} Beyond facility reallocation~\cite{de_keijzer_facility_2022,FOTAKIS202113}, there is also a line of work on online learning for clustering, where the requests are revealed \emph{after} the centers are chosen. In this line, the movement is usually free, and regret guarantees are against an optimal \emph{static} set of centers~\cite{Fotakis2021EfficientOL, pmlr-v130-cohen-addad21a,hebbar2025learning}. With movement costs, \cite{Christou23Efficient} studies the $k$-median objective, but still compares against an offline benchmark with fixed centers. In our setting, the requests are revealed before the configuration is selected, and we focus on the competitive ratio, which compares against an optimal offline algorithm that could move (and pay movement costs) as we do.

Another closely related line of work focuses on \emph{consistent $k$-clustering}~\cite{fichten21consistent,CompetitiveConsistent,OnlineClusteringConsistency,pmlr-v130-guo21a,LackiHGJR24,ForsterS25}. There, a set of $k$ centers must be maintained over a stream of insertions and deletions of requests such that at any time $t$, the centers form a (constant) approximation of the optimal clustering of the current requests. A main difference is that the movement is free, up to a restriction on the number of changes of the centers (usually $\widetilde{O} (k)$ in total when requests only arrive and $O(1)$ per request when requests could arrive or leave).

\paragraph{Facility Location.} When there is no fixed number of centers but one can open centers at some cost (and cannot move), the problem becomes the classical facility location problem, which has been studied with both adversarial order~\cite{Fotakis07,Fotakis08,DBLP:conf/focs/Meyerson01, NAGARAJAN2013361} and random-order inputs~\cite{Kaplan2023AlmostTB, DBLP:conf/focs/Meyerson01}. The variant where the connection cost is the sum of radii of facilities is considered in~\cite{Fotakis2011OnlineSC}, similar to our norm-of-radii objective.

A related variant is {facility location in evolving metrics}~\cite{AnNS17,EisenstatMS14}, in which the metric space changes over time. Instead of paying to move facilities, the algorithm may reassign requests to facilities, incurring reassignment costs. The online version of the problem is studied in~\cite{FotakisKZ21}.

\paragraph{$k$-Server and Variants.} The $k$-server problem is one of the most prominent problems in online algorithms. Deterministically, the competitive ratio for $k$-server is $k$ \cite{MMS88, CoesterKZ26}. With randomization, it is possible to achieve a competitive ratio of $O(\log^2 k)$ on HSTs~\cite{Bubeck18}, which implies an $O(\log^2k \log n)$-competitive algorithm on any $n$-point metric. Very recently, OpenAI announced an $O(\log^2k)$-competitive algorithm against oblivious adversaries~\cite{OAI26}, matching the $\Omega(\log^2 k)$ lower bound~\cite{BubeckCR23}. The reader is referred to \cite{KOUTSOUPIAS09kserver} for a survey on the $k$-server problem, and~\cite{BubeckCR23,CoesterKZ26} for more recent results.

Many variants of the $k$-server problem have been studied, including the $(h,k)$-server problem~\cite{BansalEJK17, BienkowskiBCJ20}, the infinite server problem~\cite{coester2017infinite}, the weighted $k$-server problem~\cite{FIAT199485, Bansal2017WeightedKB, Bijoy2025}, the generalized $k$-server problem~\cite{KT04, DBLP:conf/soda/Bansal0KN18, Sitter06generalized}, and the $k$-taxi problem~\cite{coester2019online, coester26taxi, anupam24taxi}.

\paragraph{Metrical Task and Service Systems.} The $k$-server problem (and all of the variants above except for $k$-taxi) belongs to the class of metrical task systems (MTS)~\cite{BorodinLS92, BubeckCR23, Coester19pure, bubeck21metrical}. Online CMC is also a metrical task system. Metrical Service Systems (MSS)~\cite{Chrobak1991TheSP} are a special case of MTS with costs in $\{0,\infty\}$, and can be viewed as a single server chasing a set of feasible locations. The works~\cite{bubeck22shortest, BURLEY1996479, fiat1998competitive, CoesterT26} study MSS with size-$w$ sets and show that the deterministic competitive ratio is $\Theta(2^w)$, while the randomized competitive ratio is $\Theta(w^2)$.

A closely related problem is metrical service systems with multiple servers (MSSMS), in which each request is a set of $w$ points and the algorithm must move so that at least one server reaches one of them. This problem is known to be difficult, and competitive algorithms are only known for the uniform metric~\cite{Chiplunkar15metrical, Feuerstein98uniform}. Although MSSMS and CMC both involve multiple servers responding to set-valued requests, the two problems capture fundamentally different objectives. In MSSMS, it suffices to serve \emph{one} point in the request set, whereas in CMC the cost is determined by the distances from the centers to \emph{all} points in the request set.

\paragraph{Online Convex Function Chasing.} Another related topic is online convex function chasing. The paper~\cite{bubeck21online} studies chasing convex functions with $k$ servers. They give strong general lower bounds and provide upper bounds for a class of functions that they call \emph{well-sharpened}.
When $k=1$, chasing general convex functions and convex bodies in $\R^d$ can be done deterministically with a competitive ratio $O(d)$~\cite{sellke20chase, ArgueGGT20}.

\section{Preliminaries}
\label{sec: prelim}
We write $[n] = \{1, \ldots, n\}$. Let $M = (X,d)$ be a metric space. For $u \in X$ and $S \subseteq X$, we write $d(u, S) = \min_{v \in S} d(u, v)$. A \emph{configuration} is a multiset of points in $X$ of size $k$.

We use the notation $(x_a)_{a \in A}$ to denote the vector with elements in $A$ as coordinates. For example, we write $(d(r, C))_{r \in R_t}$ where $C$ is a configuration to denote the vector of the distances from each request in $R_t$ to its nearest center in $C$. For a vector-valued function $f\colon\R\to\R^A$, we use $f_a(\tau)$ to denote its coordinate corresponding to $a \in A$.

We use the following definitions related to norms:
\begin{enumerate}[itemsep=0em, topsep=0.3em]
    \item A norm $\|\cdot\|$ is \emph{monotone} if for all $0 \le u \le v$, $\|u\| \le \|v\|$.
    \item A norm $\|\cdot\|\colon \R^m \to \R$ is \emph{symmetric} if for any $x \in \R^m$ and permutation $\pi\colon[m]\to [m]$,
\[
\|(x_1, \ldots, x_m)\| = \|(x_{\pi(1)}, \ldots, x_{\pi(m)})\|.
\]
    \item The dual norm of a norm $\|\cdot\|\colon \R^m \to \R$ is defined as \[
        \|z\|_* = \max\{ |\langle x, z \rangle| \colon \|x\| = 1\}.
    \]
    It is well-known that the dual of the dual norm is the original norm.
Any $z$ such that $\|z\|_* = 1$ and $|\langle z, x \rangle| = \|x\|$ is called a \emph{norming functional} for $x$.
\end{enumerate}
We will use the following property on the non-negativity of the norming functional, whose proof appears in Appendix~\ref{appendix: prelim}.
\begin{lemma}
\label{lemma: positive dual}
    If a norm $\|\cdot\|\colon\R^m \to \R$ is monotone, then for any $x \in \R^m_{\ge 0}$, there exists $z \in \R_{\ge 0}^m$ such that $\|z\|_* = 1$ and $\langle z,x\rangle = \|x\|$.
\end{lemma}

We will consider two classes of clustering costs. In the \emph{norm-of-distances} model, the cost is a norm of the vector of distances from each request to its nearest center: \[ \Clust_t(R_t, C_t) = \bigl\|\bigl( d(r, C_t) \bigr)_{r\in R_t}\bigr\|, \] where $\|\cdot\|$ is a norm that is monotone.

In the \emph{norm-of-radii} model, the cost is a norm of the radii of the clusters. Let $\rad \in \R^k$ be defined such that $\rad_i$ is the radius of the $i$-th cluster (i.e., the distance from the center to the farthest point in the cluster); then
\[ \Clust_t(R_t, C_t) = \bigl\|\rad\bigr\|, \] where $\|\cdot\|$ is a norm that is symmetric and monotone.

As mentioned in the Introduction, we measure performance by the competitive ratio against an adversary who generates the requests and also serves them. For a request sequence $\sigma$, let $\mov_t(\sigma)$, $\clust_t(\sigma)$, and $\cost_t(\sigma)$ denote the movement, clustering, and total cost incurred by our online algorithm in round $t$; without the subscript, these denote the total over all rounds. We use a star to denote the corresponding adversary's costs, with arbitrary but fixed tie-breaking; for example, $\mov^*(\sigma)$, $\clust^*(\sigma)$, and $\cost^*(\sigma)$. We may omit $\sigma$ when it is clear from the context. A randomized algorithm is $\alpha$-competitive against a given adversary if, for every request sequence $\sigma$ generated by that adversary, $\E[\cost(\sigma)] \le \alpha \cdot \E[\cost^*(\sigma)] + \beta$, where $\beta$ is a constant independent of $\sigma$. Note that the cost of the adversary can be a random variable if the adversary adapts to a randomized online algorithm.

\section{The Norm-of-Distances Objective}
\label{sec: norm}
\newcommand{\shadowserv}{\textsc{ShadowServe$(\mc A)$}\xspace}
\newcommand{\farthestserv}{\textsc{FarthestServe$(\mc A)$}\xspace}

In this section, we prove \cref{thm:any-norm-intro} by showing how to convert a $k$-server algorithm into an online CMC algorithm with the same (asymptotic) competitive ratio.

\subsection{Warm-Up: The $k$-Center Objective}
Before proceeding to the main algorithm, we begin with a simple deterministic reduction for the case where the norm in each round is $\ell_\infty$, corresponding to the $k$-center objective. In this case, the clustering cost in round $t$ is $\max_{r\in R_t} d(r,C_t)$, the maximum distance from a request in $R_t$ to its nearest center in $C_t$. Let $\mc A$ be a competitive $k$-server algorithm. In each round, our algorithm \farthestserv first outputs the current configuration and thus incurs its clustering cost, then selects a farthest request, issues it as a $k$-server request to the underlying algorithm $\mc A$, and follows the configuration of $\mc A$. The pseudocode is given in \cref{alg:shadowserv-center}.
\begin{algorithm}[h]
\caption{\farthestserv}
\label{alg:shadowserv-center}
\begin{algorithmic}[1]
\Require A $k$-server algorithm $\mc A$ on $(X,d)$, initial configuration $C_0$
\State Let $S_t$ denote the configuration of $\mc A$ at round $t$, initially $S_0 \gets C_0$.
\For{rounds $t=1,2,\dots$}
  \State Observe request multiset $R_t \subseteq X$.
  \State Output $C_t \gets S_{t-1}$. \Comment{Pay the clustering cost before movement}
  \State $r_t \gets \arg\max_{r \in R_t} d(r, C_t)$ \Comment{The farthest request, breaking ties arbitrarily}
  \State Issue a $k$-server request $r_t$ to $\mc A$, resulting in $S_t$.
\EndFor
\end{algorithmic}
\end{algorithm}

Intuitively, since the $\ell_\infty$-norm clustering cost is determined solely by a single request, we can view the round as effectively containing only the farthest request and reducing to the case $|R_t|=1$. In this regime, treating a singleton request as a $k$-server request loses only a constant factor (cf. \cite{bubeck21online}). Moreover, issuing only the farthest request from the \emph{online} configuration gives a valid lower bound on the offline clustering cost, since the request farthest from the offline configuration is, by definition, at least as far from the offline configuration as the issued request.

\begin{theorem}
\label{thm: furtherserv}
    If $\mc A$ is a deterministic $\alpha$-competitive algorithm for the $k$-server problem, then $\farthestserv$ is a deterministic $O(\alpha)$-competitive algorithm for online CMC in the norm-of-distances model when the norm is $\ell_\infty$.
\end{theorem}
\begin{proof}
    We write $M^\mc A$ for the total movement of the $k$-server algorithm $\mc A$ (on the generated request sequence) and $M^*$ for the movement of the offline $k$-server algorithm. Recall that $\cost$ refers to that of \farthestserv and $\cost^*$ refers to the cost of the offline CMC algorithm.

    We first bound the $\cost$ of \farthestserv by the movement of the $k$-server algorithm $M^\mc A$: for each round $t \ge 1$ (for completeness, let $S_{-1} = S_0$),
    \begin{align}
        \mov_t &= \Match(C_{t-1}, C_t) = \Match(S_{t-2}, S_{t-1}), \label{eq:farthest-1}\\
        \clust_t &= d(r_t, C_t) = d(r_t, S_{t-1}) \le \Match(S_{t}, S_{t-1}), \label{eq:farthest-2}
    \end{align}
    where equation~\eqref{eq:farthest-1} follows from the algorithm specification that $C_t = S_{t-1}$, and inequality~\eqref{eq:farthest-2} holds because serving $r_t$ from $S_{t-1}$ costs at least $d(r_t,S_{t-1})$. Thus,
    \[
    \cost = \sum_t (\mov_t + \clust_t) \le \sum_t \Match(S_{t-2}, S_{t-1}) + \sum_t \Match(S_{t-1}, S_{t}) \le 2 M^{\mc A}.
    \]

    On the offline side, let $C_t^*$ be the configuration of an optimal CMC solution in round $t$. Consider the following offline $k$-server strategy: at the beginning of round $t$, move to $C_t^*$; then, send the closest server to the request in round $t$ and return it to its original position. The movement between rounds is exactly the same as in the CMC solution, and in round $t$ the request $r_t$ incurs a movement equal to $2d(r_t,C_t^*)$. Since the optimal offline $k$-server cost is no larger than the cost of this particular strategy, we obtain
    \begin{align*}
         M^* &\le \sum_t (\mov^*_t + 2 d(r_t, C_t^*)) \\&\le \sum_t (\mov^*_t + 2\max_{r \in R_t} d(r, C_t^*)) = \sum_t (\mov^*_t + 2\clust_t^*) \le 2\cdot \cost^*.
    \end{align*}
    Since $\mc A$ is $\alpha$-competitive, meaning $M^\mc A \le \alpha \cdot M^* + \beta$, we conclude
    \[
    \cost \le 2\cdot M^{\mc A} \le 2 \alpha \cdot M^* + 2\beta \le 4\alpha \cdot \cost^* + 2\beta. \qedhere
    \]
\end{proof}
Taking $\mc A$ to be the $k$-competitive Work Function Algorithm (WFA)~\cite{KP95,CoesterKZ26}, we achieve a deterministic $O(k)$-competitive algorithm.

\subsection{The Main Algorithm}
Our algorithm for more general norms follows the same basic framework as in the previous section: we seek to construct a $k$-server request sequence such that (i) the cost of the online CMC algorithm is \emph{upper bounded} by that of the online $k$-server algorithm, and (ii) the cost of the offline CMC algorithm is \emph{lower bounded} by that of the offline $k$-server algorithm.

However, \farthestserv relies heavily on the fact that the $\ell_\infty$-norm is determined by a single request. For a general norm, issuing just one request is no longer sufficient. Instead, we aim to generate a sequence of requests in a way that captures the contributions of the different distances. To achieve this, we use a norming functional to linearize the norm at the current distance vector. By Lemma \ref{lemma: positive dual}, for any vector $x\ge 0$ and norm $\|\cdot\|$ that is monotone, there is a nonnegative norming functional $z\ge 0$ satisfying $\|z\|_* = 1$ and $\langle z,x\rangle = \|x\|$. The norming functional $z$ provides a principled way to decompose the norm into contributions from the different requests. Since $z \ge 0$, we can interpret its coordinates as sampling rates in Poisson processes and use them to generate a random sequence of $k$-server requests. The pseudocode of our algorithm is given in \Cref{alg:shadowserv}.

\begin{algorithm}[h]
\caption{\shadowserv}
\label{alg:shadowserv}
\begin{algorithmic}[1]
\Require A $k$-server algorithm $\mc A$ on $(X,d)$, initial configuration $C_0$
\State Let $S(\tau)$ denote the configuration of $\mc A$ at time $\tau \in \R$, initially $S(0) \gets C_0$
\For{rounds $t=1,2,\dots$}
  \State Observe request multiset $R_t \subseteq X$ and the norm $\|\cdot\|$ on $\mathbb{R}^{R_t}$
  \ForAll{infinitesimal time intervals $[\tau,\tau + \d\tau) \subseteq (t-1,t]$}
    \State Let $d(\tau) \gets \big(d(r,\, S(\tau))\big)_{r\in R_t}$ \Comment{Distances to nearest server}
    \State Choose a norming functional $z(\tau)\in \mathbb{R}^{R_t}_{\ge 0}$ s.t.
      $\|z(\tau)\|_* = 1$ and $\langle z(\tau), d(\tau)\rangle = \|d(\tau)\|$.
      \ForAll{requests $r \in R_t$ \textbf{independently}}
      \State Issue a $k$-server request $r$ to $\mc A$ with probability $z_r(\tau)\, \d\tau$

    \EndFor
  \EndFor

  \State Select $\tau_t \in \arg\min_{\tau\in(t-1,t]} \|d(\tau)\|$ \Comment{Optimal clustering time}
  \State Set $C_t \gets S(\tau_t)$
\EndFor
\end{algorithmic}
\end{algorithm}
\paragraph{The Algorithm.}  We run the $k$-server algorithm $\mc A$ in continuous time and generate a \emph{shadow} $k$-server request process from the CMC instance. Discrete round $t$ of the CMC instance corresponds to the time interval $(t-1,t]$ of the $k$-server instance. For $\tau \in \R_{\ge 0}$, let $S(\tau)$ be $\mc A$'s configuration and define the distance vector $d(\tau)=(d(r, S(\tau)))_{r\in R_t}$ and its norming functional $z(\tau)\ge 0$. At each time $\tau\in(t-1,t]$, we generate shadow requests using independent Poisson clocks: for each $r\in R_t$, an arrival at $r$ occurs in $[\tau,\tau+\d\tau)$ with probability $z_r(\tau)\,\mathrm d\tau$ (equivalently, at rate $z_r(\tau)$). Algorithm $\mc A$ serves this random shadow request stream, producing the trajectory $S(\tau)$. Our CMC algorithm then outputs $C_t=S(\tau_t)$, where $\tau_t\in(t-1,t]$ minimizes $\|d(\tau_t)\|$ over the interval. Note that the configuration of the $k$-server algorithm is not rewound to $S(u_t)$ but remains $S(t)$ before serving the next round.

For intuition: when $\|\cdot\|$ is the $\ell_\infty$-norm, only the farthest request is issued once in expectation, and we behave similarly to \farthestserv; when $\|\cdot\|$ is the $\ell_1$-norm, we have $z(\tau)=\mathbf 1$, so the algorithm serves an i.i.d. sequence with unit rates for each request. The resulting trajectory tends to explore the batch evenly and leads to a configuration that is ``good on average.''

We make a few remarks about the algorithm. The Poisson process is well-defined because the rates are piecewise constant, and it can be implemented by sampling waiting times from exponential distributions. Since the total rate $\sum_{r \in R_t} z_r(u)$ is finite throughout each interval $(t-1,t]$, the Poisson process is non-explosive on finite time intervals, and simultaneous arrivals occur with probability zero. Finally, it is plausible (and acceptable) that a request from the original request set appears multiple times or does not appear at all in the shadow process during the interval.

\paragraph{The Poissonization.}
The purpose of Poissonization is to couple the clustering cost to the expected movement cost of the induced $k$-server process for both online and offline algorithms, just as in the $\ell_\infty$ case.

On the online side, at time $\tau$, each request $r$ is generated as a shadow arrival with instantaneous rate $z_r(\tau)$, where $z(\tau)$ is a norming functional of $d(\tau)$. Whenever $r$ is issued, $\mc A$ must incur movement of at least $d_r(\tau)$ to serve it. Therefore, the expected instantaneous movement cost of $\mc A$ at time $\tau$ is at least $\sum_{r\in R_t} z_r(\tau)\,d_r(\tau)\;=\;\|d(\tau)\|$, which is exactly the instantaneous clustering cost.

On the offline side, consider the offline $k$-server algorithm that follows the offline CMC solution and serves each request with the closest center and moves back. Denote by $d^*$ the offline distance vector and $z^*$ its norming functional. The expected within-round cost of this algorithm is $2\sum_{r\in R_t} z_r(\tau)d_r^\ast =2\langle z(\tau), d^* \rangle\le 2 \langle z^*, d^* \rangle = 2\|d^*\|$,
since the shadow requests are generated using the \emph{online} norming functional $z(\tau)$ rather than one aligned with $d^*$.

We now formalize the reasoning above.

\begin{theorem}\label{thm:any-norm}
If $\mc A$ is a randomized $\alpha$-competitive algorithm for the $k$-server problem against adaptive online adversaries, then $\shadowserv$ is a randomized $O(\alpha)$-competitive algorithm for online CMC against adaptive online adversaries.
\end{theorem}

\begin{proof}
Fix an adaptive online adversary for online CMC, and consider an adaptive online adversary for $k$-server that generates the request sequence as in our reduction. Let $M^*$ denote the movement cost of the $k$-server adversary. Let $M_t^{\mc A}$ be the movement cost of $\mc A$ during round $t$ against the generated sequence, and let $M^\mc A$ be the overall movement cost of $\mc A$.

\begin{claim}
\label{claim:alg-mov-server-mov}
    The movement cost of \shadowserv is bounded by the movement cost of the $k$-server algorithm $\mc A$: in every realization, \[\mov \le M^{\mc A}.\]
\end{claim}
\begin{proof}
Fix a realization of the randomness. Let $\tau^{(1)}, \tau^{(2)},\ldots$ be the sequence of times at which $\mc A$ moves, and let $u_0=u^{(0)}=0$. By the triangle inequality,
\begin{align*}
    \mov =\sum_t \Match(C_{t-1}, C_{t}) &= \sum_t \Match(S(\tau_{t-1}), S(\tau_{t}))
    \\&\le \sum_{t} \sum_{\tau^{(i)} \in (\tau_{t-1}, \tau_{t}]} \Match(S(u^{(i-1)}), S(u^{(i)})) \\&\le \sum_t M_t^{\mc A} = M^{\mc A}.\qedhere
\end{align*}
\end{proof}

\begin{claim}
\label{claim:alg-serv-server-mov}
    The expected clustering cost of \shadowserv is bounded by the expected movement cost of the $k$-server algorithm $\mc A$: \[\E[\clust] \le \E[M^\mc A].\]
\end{claim}
\begin{proof}
For each round $t$,
\begin{align}
\E[\|d(\tau_t)\|]
&\le \E\left[\int_{t-1}^{t} \|d(\tau)\| \, \d\tau\right] \label{eq:shadow-1} \\
&= \E\left[\int_{t-1}^{t} \langle z(\tau), d(\tau)\rangle \, \d\tau \right]\label{eq:shadow-2} \\
&= \int_{t-1}^{t} \E[\langle z(\tau), d(u)\rangle] \, \d\tau \label{eq:shadow-3}\\
&\le \E[M_t^{\mc A}] \label{eq:shadow-4},
\end{align}
where inequality~\eqref{eq:shadow-1} holds because we choose $\tau_t$ with the smallest clustering cost, equality~\eqref{eq:shadow-2} follows from the definition of $z$, and equality~\eqref{eq:shadow-3} follows from Fubini's theorem, which allows swapping the integration order. For inequality~\eqref{eq:shadow-4}, at any time $\tau\in (t-1,t]$, each request $r\in R_t$ is issued independently at instantaneous rate $z_r(\tau)$, and if it is issued then $\mc A$ must move by at least $d(r,S(\tau))$ to serve it. Hence the expected infinitesimal movement cost of $\mc A$ at time $\tau$ is at least $\E[\sum_{r\in R_t} z_r(\tau)\, d(r,S(\tau))]
= \E[\langle z(\tau), d(u)\rangle]$, and the inequality follows from integration.

The claim then follows from summing over $t$.
\end{proof}

\begin{claim}
\label{claim:server-opt-facility-opt}
    The expected cost of the adaptive online adversary for the $k$-server algorithm is bounded by the expected cost of the adaptive online adversary's CMC algorithm:
    \[
    \E[M^*] \le 2\cdot \E[\cost^*].
    \]
\end{claim}
\begin{proof}
    Let $C_t^*$ denote the (random) configuration of an adaptive online adversary for online CMC at round $t$ and let $d^*_t = (d(r, C_t^*))_{r \in R_t}$ be the corresponding distance vector. Consider the following service strategy of the adaptive online $k$-server adversary: at round $t$, first move the configuration to $C_t^*$, then whenever a request during round $t$ appears, move the closest server to the request and move it back.

    The expected cost of this strategy, conditioned on $d_t^*$, during round $t$ is
    \begin{align*}
        2 \int_{t-1}^{t} \E[\langle z(\tau), d^*_t\rangle \mid d_t^*] \, \d\tau
    \le2 \int_{t-1}^{t} \left(\max_{\|z\|_*=1} |\langle z, d^*_t \rangle|\right)\d\tau
    =
    2 \int_{t-1}^{t} \|d^*_t\| \, \d\tau
    =
    2\cdot \clust^*_t.
    \end{align*}

    Further, this strategy moves between rounds exactly as the adaptive online adversary for online CMC does. Thus, the total cost of this adaptive online adversary for $k$-server is upper-bounded by
    \[\sum_t (2\cdot \E[\clust^*_t] + \E[\mov^*_t]) \le 2\cdot \E[\cost^*]. \qedhere\]
\end{proof}

Combining Claims \ref{claim:alg-mov-server-mov}, \ref{claim:alg-serv-server-mov}, and \ref{claim:server-opt-facility-opt} gives that our total expected cost satisfies
\[
\E[\cost] = \E[\mov] + \E[\clust] \le 2 \E[M^{\mc A}] \le 2\alpha \cdot \E [M^*] + 2\beta \le 4\alpha \cdot \E[\cost^*] + 2\beta. \qedhere
\]

\end{proof}

Taking $\mc A$ to be the $k$-competitive Work Function Algorithm~\cite{KP95,CoesterKZ26} yields \cref{thm:any-norm-intro}.

We remark that the reduction does not directly apply to algorithms that are competitive \emph{only} against \emph{oblivious} adversaries. Indeed, the shadow requests are generated at rates that depend on the current configuration of the $k$-server algorithm, so the resulting request stream is adaptive from the perspective of $\mc A$, and accordingly it will not suffice if $\mc A$ is competitive only against oblivious adversaries. Nevertheless, when the norming functional is constant, which is the case under weighted $\ell_1$-norms (capturing $k$-median and generalizations), the shadow-request rates are fixed independently of $\mc A$, and the reduction also applies to oblivious-adversary algorithms.

\subsection{Generalization to Chasing Problems}
The Poissonization approach can be generalized beyond clustering: it provides a generic way to \emph{linearize} a batched, norm-aggregated distance cost into the movement of a shadow online process. In this section, we instantiate this idea for chasing problems~\cite{Chrobak1991TheSP, Friedman93convex} (also known as Metrical Service Systems).

\begin{definition}[Batched Chasing Problem]
    An $\mc F$-chasing problem is defined on a metric space $(X,d)$ with an initial
    location $x_0$, where $\mc F \subseteq 2^X$ specifies the possible input sets. At each time $t$,
    the algorithm receives a set $A_t \in \mc F$ and must output $x_t \in A_t$. The total cost of the algorithm is
    $\sum_{t\ge1} d(x_{t-1},x_t)$.

    A batched $\mc F$-chasing problem is defined on a metric space $(X,d)$ with an
    initial location $x_0$, where $\mc F \subseteq 2^X$ specifies the possible input sets. At each time
    $t$, the algorithm receives a multiset of feasible sets $R_t \subseteq \mc F$ and a monotone norm
    $\|\cdot\|:\R^{R_t}\to \R$, and must output $x_t \in X$. The total cost of the algorithm is
    \[
        \sum_{t\ge 1} \bigg(d(x_{t-1},x_t) + \|(d(x_t,F))_{F\in R_t}\|\bigg).
    \]
\end{definition}

The batched model is motivated by settings in which many feasibility constraints or preferred regions arise simultaneously, for example, from multiple clients, sensors, tasks, or ``local minima'', while movement can occur only infrequently. Since the intersection of the regions may be empty, it is natural to use a ``soft'' penalty: incur the distance to each requested set and aggregate these distances using a norm.

\begin{example}[$k$-Server and CMC]
    The $k$-server problem on $X$ can be viewed as an $\mc F$-chasing problem on the configuration space $X \choose k$, with
    \[
    \mc F = \left \{A_x : x\in X\right\}, \text{ where } A_x := \left \{C\in {X \choose k} : x \in C\right \}.
    \]
    Upon receiving a request $x_t$, the algorithm must move to a configuration in $A_{x_t}$. The batched version is then the online CMC problem.
\end{example}

\begin{example}[Convex Body Chasing]
    In convex body chasing, the metric space is $\R^d$ with some norm and $\mc F$ contains all convex sets in $\R^d$. The problem is known to admit a deterministic $O(d)$-competitive algorithm~\cite{sellke20chase, ArgueGGT20}.

    As an example of the batched version, consider a robot choosing a location at each step. Each subsystem (e.g., safety, latency, energy, regulation) outputs a convex recommended region. The batched soft cost measures how badly each subsystem is violated, together with natural movement costs. Note that the norm allows different coefficients for different subsystems to reflect their relative importance.
\end{example}

\begin{example}[Small set chasing]
In small set chasing, the metric space $(X,d)$ is arbitrary and $\mc F$ contains all subsets of cardinality at most $w$. The problem is known to admit a randomized $O(w^2)$-competitive algorithm against oblivious adversaries~\cite{bubeck22shortest} and an $O(2^w)$-competitive deterministic algorithm~\cite{CoesterT26}.

The batched version models scenarios where at each time step one has to choose an operating location/configuration (e.g., docking station for a drone, or base location chosen by a mobile technician each day), and tasks/clients can only be served from a small set of feasible options (e.g., a limited number of access points).
\end{example}

An adaptation of \textsc{ShadowServe}, which now requires an $\mc F$-chasing algorithm, yields the more general reduction to batched $\mc F$-chasing problems (\cref{thm: mss}). The proof is similar to that of \cref{thm:any-norm}, and we defer it to Appendix~\ref{appendix:mss}.

\begin{corollary}
\label{cor:convex body}
    There is an $O(d)$-competitive deterministic\footnote{For convex body chasing, randomization does not help since an algorithm can always go to its expected location. See Appendix~\ref{appendix: derandomize} for a complete proof.} algorithm for batched convex body chasing.
\end{corollary}

\begin{corollary}
    For batched small set chasing, there is a randomized $O(2^w)$-competitive algorithm against adaptive online adversaries. If $\|\cdot\|$ is a weighted $\ell_1$-norm, then there is a randomized $O(w^2)$-competitive algorithm against oblivious adversaries.
\end{corollary}

Batched $\mc F$-chasing is an instance of a metrical task system (MTS)~\cite{BorodinLS92}. In full generality, MTS is quite difficult: the optimal competitive ratio on $n$-point metrics is $\Theta(n)$ for deterministic algorithms and $\Theta(\log^2 n)$ for randomized algorithms against oblivious adversaries, which become vacuous on unbounded or continuous metric spaces. Much of the literature on concrete problems (such as the $k$-server problems) can be viewed as identifying highly structured subclasses that admit better guarantees (cf.~\cite{bubeck20parametrize}). \Cref{thm: mss} offers one perspective on this direction by showing that once a problem is ``well-behaved'', its batched, norm-aggregated counterpart remains well-behaved. This perspective shares the spirit of \cite{bubeck21online} on viewing multi-server convex function chasing as non-convex optimizations (on $\R^{kd}$) with some inherent convexity.

\section{Deterministic Algorithms on the Line}
\label{sec: line det}
In this section, we prove \cref{thm:aggrdc-intro} and \cref{thm:radii-intro} by providing \emph{deterministic} $O(k)$-competitive algorithms for both the norm-of-distances and norm-of-radii objectives on the line. Note that there is an $\Omega(k)$ lower bound for deterministic algorithms, inherited from the lower bound for the $k$-server problem~\cite{MMS88}.

Our algorithm for the norm-of-distances objective can be viewed as a derandomization of the ideas from the previous section, where Poissonization is replaced by a linear combination of different movements whose speeds are scaled according to the norming functional. For the norm-of-radii objective, a new complication is a double-charging problem: Several online centers could charge their cost to the same offline center, which naively could harm the competitive ratio. To avoid this, we exploit the geometry of the line and use symmetry of the norm to redistribute charges in a way that limits the amount of double-charging. Our algorithms will adapt the Double Coverage algorithm for the $k$-server problem, and aggregate the movements with the norming functional that allows us to relate the movement cost and the clustering cost.

It may help to briefly recall the Double Coverage algorithm for the $k$-server problem on the line~\cite{CKP91new}. Upon a request, the algorithm moves the adjacent server(s)---two if the request lies between two current servers, and one if all servers lie on the same side of the request---towards the request at the same speed, and stops as soon as a server reaches the request location.

Denote by $x_1\le x_2\le\dots\le x_k$ the locations of the online servers, sorted from left to right, and by $y_1\le y_2\le\dots\le y_k$ the locations of the offline servers. The analysis uses the potential function
\[
\Phi \;=\; M+\Sigma
\]
where $M=k \sum_i |x_i-y_i|$ and $\Sigma \;=\; \sum_{i < j}|x_i-x_j|$. When the offline solution moves, we have $\Delta \Phi \le k\cdot \mov^*$. When the online solution moves, we use that there is an offline server at the requested location. If the algorithm moves two servers by a distance $L$ each, then at least one of them moves toward its matched partner, so $\Delta M \le 0$ and $\Delta \Sigma =- 2L$, and hence $\Delta \Phi + \mov \le 0$. If the algorithm moves only one server by a distance $L$, then it moves toward its matched partner, so $\Delta M = -k L$ and $\Delta \Sigma = (k-1)L$, and again $\Delta \Phi + \mov \le 0$. This completes the proof that Double Coverage is $k$-competitive~\cite{CKP91new}.

\subsection{Norm-of-Distances Objective}
\newcommand{\AggrDC}{\textsc{AggrDC}\xspace}

\begin{algorithm}[t]
\caption{\AggrDC}
\label{alg:aggrdc}
\begin{algorithmic}[1]
\Require initial configuration $C_0$ (multiset of size $k$)
\State Initialize $S(0) \gets C_0$
\For{rounds $t=1,2,\dots$}
  \State Observe request multiset $R_t$ and the norm $\|\cdot\|$ on $\mathbb{R}^{R_t}$
  \For{continuous time $\tau \in (t-1,t]$}
    \State Let $d(\tau) \gets \big(d(r,\, S(\tau))\big)_{r \in R_t}$
    \State Choose $z(\tau)\in \mathbb{R}^{R_t}_{\ge 0}$ s.t.
      $\|z(\tau)\|_* = 1$ and $\langle z(\tau), d(\tau)\rangle = \|d(\tau)\|$
    \For{$r \in R_t$} \Comment{Compute single-request Double Coverage velocities}
      \State $v_r(\tau) \gets \mathbf 0$ \Comment{Velocity induced by request $r$}
      \If {there exists a center $i$ that is the closest center on the left of $r$}
          \State $v_{r, i}(\tau) \gets d_r(\tau)$ \Comment{Left center moves right}
      \EndIf
      \If {there exists a center $i$ that is the closest center on the right of $r$}
          \State $v_{r,i}(\tau) \gets -d_r(\tau)$ \Comment{Right center moves left}
      \EndIf
    \EndFor
    \State $V(\tau) \gets \sum_r z_r(\tau) v_r(\tau)$ \Comment{Aggregate using norming functional}
    \State $S(\tau+\d\tau) \gets S(\tau) + V(\tau) \d\tau$ \Comment{Move according to velocity $V$}
  \EndFor
  \State Select $\tau_t \in \arg\min_{\tau\in(t-1,t]} \|d(\tau)\|$
  \State Set $C_t \gets S(\tau_t)$
\EndFor
\end{algorithmic}
\end{algorithm}

We first focus on the norm-of-distances objective. Similar to Section \ref{sec: norm}, we will work in continuous time $\tau \in \R_{\ge 0}$ with the requests $R_t$ taking effect during $\tau \in (t-1, t]$. We specify continuous configurations $S(\tau)$, and the discrete configuration $C_t$ is determined by the configuration with the smallest clustering cost in $(t-1, t]$.

At a high level, we move adjacent center(s) to each request as in Double Coverage, but scale the speed for movement towards any request by its current contribution to the total clustering cost. Concretely, for $\tau \in (t-1, t]$, let $d(\tau) = (d(r, S(\tau)))_{r \in R_t}$ and let $z(\tau)$ be the norming functional for $d(\tau)$ (with respect to the given norm $\|\cdot\|$). For each request $r$, we move adjacent centers of $r$ towards it at a base speed proportional to $d_r(\tau)$. A double coverage analysis shows that either we improve our potential, or there is a lower bound on the offline clustering cost of $r$ enough to pay for our clustering cost of $r$. The final velocity $V(u)$ is the accumulation of movements induced by requests, scaled by $z(\tau)$, which captures the contribution of $r$ at that time. The pseudocode is given in \cref{alg:aggrdc}.

In the main text, we present the movement rule and analysis in continuous time for ease of understanding. Appendix~\ref{appendix:discretize} discusses finite implementations with the same asymptotic competitive ratio.

\begin{theorem}
\label{thm:aggrdc}
Algorithm \textsc{AggrDC} is $O(k)$-competitive on the line.
\end{theorem}

\begin{proof}
We first prove that the continuous time algorithm, which moves and clusters in continuous time (with requests in $R_t$ extended to $\tau \in (t-1, t]$), is $O(k)$-competitive, i.e.,
\[
\cost^{\text{cont}}  := \int (\|V(\tau)\|_1 + \|d(\tau)\|) \mathrm \d\tau \le O(k)\cdot \cost^*.
\]

We adapt the Double-Coverage analysis via the potential
\[
    \Phi(\tau) \;:=\; 2\bigl(k\cdot \Match\bigl(S(\tau), C^*(\tau)\bigr) \;+\; \Sigma\bigl(S(\tau)\bigr)\bigr),
\]
where $C^*(\tau) = C^*_{\lceil \tau \rceil}$ is the optimal offline configuration extended to continuous time, and $\Sigma(S)$ is the sum of pairwise distances between %
online centers.%

When the offline solution moves, $\Sigma$ is unchanged, and $\Match$ increases by at most the offline movement; hence, the total increase in $\Phi$ due to offline moves is at most $O(k)\,\mov^*$.

Fix the round $t$ and consider $\tau\in(t-1,t]$. Let $d^*(\tau):=\bigl (d(r, C^*(\tau))\bigr)_{r \in R_t} \in \R^{R_t}_{\ge 0}$. Recall that $v_r(\tau)$ is the velocity induced by request $r$ and let $\phi'_r(\tau)$ be the corresponding instantaneous change in the potential (not scaled by $z_r(\tau)$).

\begin{claim}
\label{claim:aggrdc-local}
For every $\tau \in (t-1, t]$ and any request $r$,
\[
\phi'_r(\tau) + \|v_r(\tau)\|_1 + d_r(\tau)\le 4k\cdot d^*_r(\tau).
\]
\end{claim}

\begin{proof}
We consider two cases based on the number of moving centers.
\begin{enumerate}
    \item {Two centers move.} In this case, $\|v_r(\tau)\|_1 = 2 d_r(\tau)$.
    If there is an offline center between them, then $\Match'_r(\tau) \le 0$ and $\Sigma'_r(\tau) = -2d_r(\tau)$, so
    \[
        \phi'_r(\tau) + \|v_r(\tau)\|_1 + d_r(\tau)\le 2(k \cdot 0 -2d_r(\tau)) + 2d_r(\tau) + d_r(\tau) \le 0.
    \]
    If there is no offline center between them, then $d_r^*(\tau) \ge d_r(\tau)$. We have $\Match'_r(\tau) \le 2d_r(\tau)$ and $\Sigma'_r(\tau) = -2d_r(\tau)$, and hence
    \[
        \phi'_r(\tau) + \|v_r(\tau)\|_1 + d_r(\tau) \le 2(k \cdot 2d_r(\tau) - 2d_r(\tau)) + 2d_r(\tau) + d_r(\tau) \le 4k \cdot d_r^*(\tau).
    \]

    \item {One center moves.} In this case, $\|v_r(\tau)\|_1 = d_r(\tau)$.
    If there is an offline center on the same side of the moving online center as the request $r$, then $\Match'_r(\tau) = -d_r(\tau)$ and $\Sigma'_r(\tau) = (k-1)d_r(\tau)$, so
    \[
        \phi'_r(\tau) + \|v_r(\tau)\|_1 + d_r(\tau)\le 2(k \cdot (-d_r(\tau))+ (k-1)d_r(\tau)) + d_r(\tau) + d_r(\tau)\le  0.
    \]
    Otherwise, all offline centers are further from $r$ than the online center, so $d_r^*(\tau) \ge d_r(\tau)$, and we have
    \[
 \phi'_r(\tau) + \|v_r(\tau)\|_1 + d_r(\tau) \le 2(k \cdot d_r(\tau) + (k-1)d_r(\tau)) + d_r(\tau) + d_r(\tau) \le 4k \cdot d_r^*(\tau). \qedhere
    \]
\end{enumerate}
\end{proof}

Let $s(\tau) := ( \|v_r(\tau)\|_1 )_{r\in R_t}$ and $\phi'(\tau):=(\phi'_r(\tau))_{r\in R_t}$.
The total amortized instantaneous cost is
\begin{align*}
\Phi'(\tau) + \|V(\tau)\|_1+\|d(\tau)\| &\le
 \langle z(\tau),\,\phi'(\tau)\rangle + \langle z(\tau), s(\tau)\rangle + \langle z(\tau),\, d(\tau)\rangle \\
&= \langle z(\tau),\, \phi'(\tau) + s(\tau) + d(\tau)\rangle \\
&\le 4k \cdot \langle z(\tau),\, d^*(\tau)\rangle \\
&\le 4k \, \|d^*(\tau)\|,
\end{align*}
where the last inequality follows from $\|z(\tau)\|_* = 1$. Integrating over time completes the proof that the continuous time algorithm is $O(k)$-competitive.

We conclude the proof by noting that
\[
\cost \le \cost^{\text{cont}}
\]
by the triangle inequality and the fact that we pick the smallest clustering cost, similar to \cref{claim:alg-mov-server-mov} and \cref{claim:alg-serv-server-mov}.
\end{proof}

\subsection{Norm-of-Radii Objective}
\label{sec: norm of radii}
\newcommand{\radAggrDC}{\textsc{RadAggrDC}\xspace}
We now present our algorithm for the norm-of-radii objective, \radAggrDC, in \Cref{alg:radaggrdc}. It is similar in spirit to \AggrDC, except that for each center we only consider its farthest request. The main differences are marked by surrounding boxes. In Lines~\ref{Line:for} and~\ref{Line:r_c}, the cluster of a center $c$ refers to all requests in $R_t$ assigned to $c$ as their closest center, breaking ties by the smallest center index (so $\rad_c$ is the maximum distance between $c$ and a request in its cluster, or zero for an empty cluster).

\begin{algorithm}[h]
\caption{\radAggrDC}
\label{alg:radaggrdc}
\begin{algorithmic}[1]
\Require initial configuration $C_0$ (multiset of size $k$)
\State Initialize $S(0) \gets C_0$
\For{rounds $t=1,2,\dots$}
  \State Observe request multiset $R_t$ and the norm $\|\cdot\|$ on $\mathbb{R}^{k}$
  \For{continuous time $\tau \in (t-1,t]$}
    \State \fbox{Let $\rad(\tau) \in \R^k$ be the vector of the radii}
    \State Choose $z(\tau)\in \mathbb{R}^{k}_{\ge 0}$ s.t.
      $\|z(\tau)\|_* = 1$ and $\langle z(\tau), \rad(\tau)\rangle = \|\rad(\tau)\|$
    \State \Comment{Norming functional for $\rad(\tau)$.}
    \For{$c \in [k]$ with a non-empty cluster} \label{Line:for}
      \Comment{Compute single-request velocities}
      \State $v_c(\tau) \gets \mathbf 0$ \Comment{Velocity induced by center $c$}
      \State \fbox{$r_c \gets$ the farthest request in the cluster of $c$}\label{Line:r_c}
      \If {there exists a center $i$ that is the closest center on the left of $r_c$}
          \State $v_{c, i}(\tau) \gets$ \fbox{$\rad_c(\tau)$} \Comment{Left center moves right}
      \EndIf
      \If {there exists a center $i$ that is the closest center on the right of $r_c$}
          \State $v_{c,i}(\tau) \gets$ \fbox{$-\rad_c(\tau)$} \Comment{Right center moves left\footnotemark}
      \EndIf
    \EndFor
    \State $V(\tau) \gets \sum_c z_c(\tau) v_c(\tau)$ \Comment{Aggregate using norming functional}
    \State $S(\tau+\d\tau) \gets S(\tau) + V(\tau) \d\tau$ \Comment{Move according to velocity $V$}
  \EndFor
  \State Select $\tau_t \in \arg\min_{\tau\in(t-1,t]} \|\rad(\tau)\|$.
  \State Set $C_t \gets S(\tau_t)$
\EndFor
\end{algorithmic}
\end{algorithm}
\footnotetext{When $\rad_c>0$, one of the selected centers is at the location of $c$.}

We remark that the algorithm is slightly informal. For example, under any fixed tie-breaking rule for choosing the farthest point $r_c$ in Line~\ref{Line:r_c}, the algorithm will oscillate and fail to be well-defined if $c$ has two equally distant farthest points (in that case, $c$ should stay put). Nevertheless, we choose to present the continuous analysis that we feel expresses the idea more cleanly. Appendix~\ref{appendix:discretize} gives a formal finite implementation by splitting the continuous interval into sufficiently short time steps. This approach avoids the issue of well-definedness and retains the asymptotic guarantee.

The main new challenge for the norm-of-radii objective is a double-charging problem. In the norm-of-distances analysis, the online cost of a request is charged towards the offline cost of the request (plus the change in potential), so the resulting offline costs are disjoint and can be accumulated (up to the norming functional). Here, however, each offline center only contributes its \emph{radius} once. Our analysis for each online center $c$ charges to the cost of the \emph{offline center} that serves the farthest request of $c$. Since many different online centers can have farthest requests served by the same offline center, these lower bounds are not naturally disjoint, and directly accumulating them may overcount the offline cost.

We resolve this by exploiting two structural properties: on the line, the geometry of clusters keeps the amount of overlap within a constant factor, and for symmetric norms, merging multiple charges into a single coordinate can only increase the norm. Together, these imply that the possible overlap only contributes a constant factor.

\begin{theorem}
\label{thm:radaggrdc}
Algorithm \radAggrDC{} is $O(k)$-competitive on the line for the norm-of-radii objective when the norm at every round $t$ is symmetric.
\end{theorem}

\begin{proof}
Similar to \Cref{thm:aggrdc}, we work in continuous time first with the same potential
\[
    \Phi(\tau) \;:=\; 2\bigl(k\cdot \Match\bigl(S(\tau), C^*(\tau)\bigr) \;+\; \Sigma\bigl(S(\tau)\bigr)\bigr),
\]
where $C^*(\tau)$ is the offline configuration with $C^*(\tau) = C^*_{\lceil \tau\rceil}$, and $\Sigma(S)$ is the sum of pairwise distances between online centers.

When the offline solution moves, $\Sigma$ is unchanged, and $\Match$ increases by at most the offline movement; hence, the total increase in $\Phi$ due to offline moves is at most $O(k)\cdot\mov^*$.

Fix the round $t$ and consider $\tau\in(t-1,t]$. Let $\rad^*(\tau)$ denote the offline radii vector. For each nonempty online cluster $c$, recall $r_c$ is its selected farthest request and let $\chi(c)$ be the offline center assigned to $r_c$; extend $\chi$ arbitrarily to empty clusters. Recall that $v_c(\tau)$ is the velocity induced by $c$ and let $\phi_c'(\tau)$ be the corresponding instantaneous change in potential (not scaled by $z_c(\tau)$).
For empty clusters, set $v_c(\tau)=\mathbf 0$ and $\phi'_c(\tau)=0$.

\newcommand{\wrad}{s}
Let the instantaneous charges be
\[
\wrad_c(\tau):=\max\bigl(0,\phi'_c(\tau)+\|v_c(\tau)\|_1+\rad_c(\tau)\bigr).
\]

\begin{claim}
\label{claim:radaggrdc-local}
For every $\tau\in(t-1,t]$ and every offline center $j$,
\[
\sum_{c\in\chi^{-1}(j)}s_c(\tau)\le16k\rad_j^*(\tau).
\]
\end{claim}
\begin{proof}
For each nonempty cluster $c$, \Cref{claim:aggrdc-local} gives
$s_c(\tau)\le4kd(r_c,C^*(\tau))\le4k\rad^*_{\chi(c)}(\tau)$.
The movement rule also gives $s_c(\tau)\le4k\rad_c(\tau)$. Hence
\[
s_c(\tau)\le4k\min\{\rad_c(\tau),\rad^*_{\chi(c)}(\tau)\}.
\]

We then bound the overlap using the cluster geometry, illustrated in~\cref{fig:illustration}.
Fix an offline center $j$ with a nonempty cluster, and let $I$ be the
convex hull of its assigned requests, so $|I|\le2\rad_j^*(\tau)$.
For each nonempty online cluster $c\in\chi^{-1}(j)$ whose center lies
in $I$ at time $\tau$, the segment from its location to
$r_c$ lies in $I$. These segments have disjoint interiors, so their radii sum to at most $|I|$.
At most two nonempty online clusters in $\chi^{-1}(j)$ have centers
outside $I$, one on each side: only the nearest online location on
either side can serve a request in $I$. Thus,
\begin{equation*}
\sum_{c\in\chi^{-1}(j)}\min\{\rad_c(\tau),\rad_j^*(\tau)\}
\le|I|+2\rad_j^*(\tau)\le4\rad_j^*(\tau).
\end{equation*}
Multiplying by $4k$ bounds the total charge as claimed.
\end{proof}

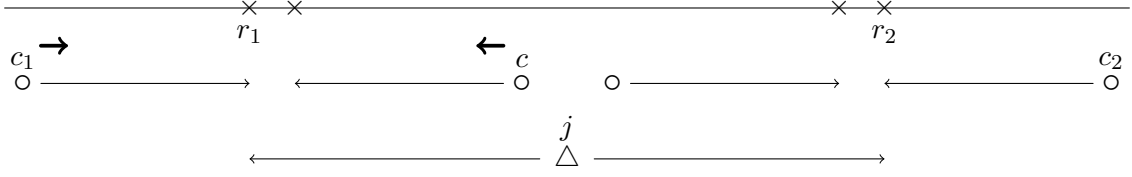
\begin{figure}[t]
  \centering

    \begin{tikzpicture}[x=1.2cm,y=1.0cm, line cap=round]
      \draw (-6.2,0) -- (6.2,0);

      \foreach \x in {-3.5,-3,3,3.5}{
        \node at (\x,0) {$\times$};
      }

      \node at (-3.5, -0.35) {$r_1$};
      \node at (3.5, -0.35) {$r_2$};

      \foreach \x in {-6,-0.5,0.5,6}{
        \node at (\x,-1) {\Large $\circ$};
      }

      \node at (-0.5, -0.7) {$c$};
      \node at (-6, -0.7) {$c_1$};
      \node at (6, -0.7) {$c_2$};

      \node at (0, -1.6) {$j$};

      \draw[line width=0.5mm,->] (-5.8,-0.5) -- (-5.5,-0.5);
      \draw[line width=0.5mm,->] (-0.7,-0.5) -- (-1,-0.5);

      \node (T0) at (0,-2) {$\triangle$};

      \draw[->] (-0.3,-2) -- (-3.5,-2);
      \draw[->] (0.3,-2) -- (3.5,-2);

      \draw[->] (-5.8,-1) -- (-3.5,-1);
      \draw[->] (-0.7,-1) -- (-3,-1);
      \draw[->] (5.8,-1) -- (3.5,-1);
      \draw[->] (0.7,-1) --(3,-1);
    \end{tikzpicture}

  \caption{Illustration for \cref{claim:radaggrdc-local}. Crosses represent requests, circles represent online centers, and the triangle represents the relevant offline center (the remaining offline centers are omitted). The arrows indicate the span of the clusters around centers. All four online centers have a farthest request served by $j$. Two lie inside the convex hull of its requests and two outside, illustrating the two contributions to the charging bound.}

  \label{fig:illustration}
\end{figure}

It remains to show that while multiple online centers may charge to the same offline center, pretending that they are charged to different offline centers provides a lower bound on the offline cost.

\begin{lemma}[Aggregation]
\label{lem:aggregation}
Let \(\|\cdot\|\) be a symmetric norm on \(\mathbb R^m\). Let
\(x\in \mathbb R^m_{\ge 0}\), and let \(\chi:[m]\to[m]\). Define a vector $y$ by $y_j := \sum_{i\in \chi^{-1}(j)} x_i$.
Then, $\|x\| \le \|y\|$.
\end{lemma}

\begin{proof}
For each \(j\) with \(y_j>0\), choose a random index
\(I_j\in \chi^{-1}(j)\) with probability distribution \[\mathbb P[I_j=i] = \frac{x_i}{y_j}.\]
Define a random vector \(Y\in\mathbb R^m\) by putting \(y_j\) at
coordinate \(I_j\) for every such \(j\), and putting \(0\) in all other
coordinates. Then \(\mathbb E [Y]=x\). Moreover, for every outcome, \(Y\)
is a permutation of \(y\). Hence, by convexity and symmetry of the norm,
\[
    \|x\|=\|\mathbb E [Y]\|
    \le \mathbb E[\|Y\|]
    = \|y\|. \qedhere
\]
\end{proof}

By \Cref{claim:radaggrdc-local}, \cref{lem:aggregation}, and monotonicity,
\[
\|s(\tau)\|\le16k\|\rad^*(\tau)\|.
\]

We are now ready to bound the total instantaneous cost.
\begin{align*}
 \Phi'(\tau) +\|V(\tau)\|_1 + \|\rad(\tau)\| &\le
\langle z(\tau),\phi'(\tau)\rangle + \langle z(\tau),(\|v_c(\tau)\|_1)_{c\in [k]}\rangle + \langle z(\tau), \rad(\tau)\rangle
\\&\le \langle z(\tau), \wrad(\tau)\rangle
\\&\le  \|\wrad(\tau)\|\\
&\le 16k \|\rad^*(\tau)\|.
\end{align*}

Integrating over time bounds the cost of the continuous time algorithm by
\[
\cost^{\text{cont}}  := \int (\|V(\tau)\|_1 + \|\rad(\tau)\|) \mathrm \d\tau \le O(k)\cdot \cost^*.
\]

We conclude the proof by noting that $\cost \le \cost^{\text{cont}}$ by the triangle inequality and the fact that we cluster at the smallest cost, similar to \cref{claim:alg-mov-server-mov} and \cref{claim:alg-serv-server-mov}.
\end{proof}

\section{Conclusion and Discussion}
\label{sec: conclusion}
In this work, we study online clustering with movement cost under two natural objectives. For the norm-of-distances objective, we establish a tight connection to the online $k$-server problem. In particular, we give a randomized $O(k)$-competitive algorithm on general metrics, substantially generalizing previous guarantees, which were known only for the $k$-median objective on the line when $k \le 2$. Our main technical tool is a Poissonization-based reduction, with rates driven by a time-evolving norming functional, which may be viewed more broadly as a way to linearize batched requests into a stream of singleton requests, and we extend it to yield a general reduction for chasing problems.

For the line metric, we further obtain deterministic $O(k)$-competitive algorithms for both the norm-of-distances and norm-of-radii objectives. An interesting direction for future work is to understand whether similar deterministic guarantees can be extended beyond the line, perhaps by adapting the Work Function Algorithm and its analysis.

Another interesting direction, specific to the norm-of-radii objective, is to work beyond symmetric norms, even on the line. As we discuss in Appendix~\ref{appendix: asymm}, this appears to require some fundamentally new ideas. The asymmetric version shares some flavor of the weighted $k$-server problem, though competitiveness for an asymmetric norm-of-radii objective would not by itself imply a competitive algorithm for weighted $k$-server, since the movement costs of different centers remain identical.

\printbibliography

\appendix
\section{Monotonicity of the Norm}
\subsection{Non-Monotone Norm Gives Unbounded Competitive Ratio}
\label{appendix:non-monotone}
Here we show that online CMC with non-monotone norms has an unbounded competitive ratio, even when $k=1$. At a very high level, this is because monotonicity ensures that composing the norm with coordinatewise convex distance functions preserves convexity, which is not true for non-monotone norms. A similar idea has been used in \cite{bubeck21online} for impossibility results for multi-server convex function chasing.

\begin{theorem}
\label{thm:non-monotone-lower}
There does not exist a competitive algorithm for online CMC with a single center if the norm is allowed to be non-monotone. This holds even for randomized algorithms against oblivious adversaries, on the line.
\end{theorem}

\begin{proof}
We reduce from the problem of chasing unions of two intervals on the line metric $[0, 1]$: The algorithm starts at a given location $y_0 \in [0, 1]$. At each round $t$, the algorithm receives
\[
U_t = [a_t,b_t]\cup[c_t,d_t]\subseteq [0,1],
\qquad a_t \le b_t \le c_t \le d_t,
\]
chooses $y_t\in U_t$, and pays the movement cost
$\sum_t |y_t-y_{t-1}|$. It is shown in~\cite[Lemma 2.3]{bubeck21online} that no randomized competitive algorithm exists for this problem against an oblivious adversary. We will show that a $C$-competitive algorithm for online CMC with non-monotone norms implies a $C$-competitive algorithm for chasing unions of two intervals, and hence no such algorithm can exist.

We will simulate any instance of chasing unions of two intervals by an online CMC instance with $k=1$. For a round $t$ with request $U_t = [a_t, b_t] \cup [c_t, d_t]$ for $a_t \le b_t \le c_t \le d_t$, let $m_t = (b_t+c_t)/2$. Define the request set $R_t := \{0,a_t,b_t,m_t,c_t,d_t, 1\}$. Let the norm $\|\cdot\|: \R_{\ge 0}^{R_t} \to \R$ be defined by\footnote{Note that $v_1+v_0=1$ always holds, but replacing $v_1+v_0$ by $1$ would violate homogeneity and make $\|v\|$ not a norm.}
\[
\|v\|
:=
\bigl|v_{a_t}+v_{d_t}-(d_t-a_t)(v_1+v_0)\bigr|
+
\bigl|v_{b_t}+v_{c_t}-2v_{m_t}\bigr|
+
\varepsilon_t \|v\|_1,
\]
where $\varepsilon_t = 2^{-t}$ is a small number that is only used to make $\|\cdot\|$ a norm (instead of a semi-norm).

Let the clustering cost at time $t$ when the center is at $x$ be
\[
\clust_t(x)
:=
\left\|
\bigl(|x-r|\bigr)_{r\in R_t}
\right\|.
\]

We note two properties of $\clust_t$.

\begin{enumerate}
    \item If $x \in U_t$, then $\clust_t(x)\le 5\varepsilon_t$. Indeed, since $x\in[a_t,d_t]$, we have
    \[
    |x-a_t|+|x-d_t|-(d_t-a_t)=d_t-a_t-(d_t-a_t)=0.
    \]
    Also, since $x\notin(b_t, c_t)$, we have
    \[
    |x-b_t|+|x-c_t|-2|x-m_t|=0.
    \]
    Both absolute-value terms vanish, and only the $\varepsilon_t\|\cdot\|_1$ term remains, which is at most $5\varepsilon_t$.

    \item If $x \notin U_t$, then $\clust_t(x)\ge 2d(x,U_t)$. To see this:
    \begin{itemize}
        \item If $x\in[0,a_t]\cup[d_t,1]$, then
        \[
        |x-a_t|+|x-d_t|-(d_t-a_t)
        = 2d(x,U_t),
        \]
        and other terms are non-negative.
        \item If $x\in[b_t,c_t]$, then
        \[
        |x-b_t|+|x-c_t|-2|x-m_t| = 2d(x,U_t),
        \]
        and other terms are non-negative.
    \end{itemize}
\end{enumerate}

Repeating the request many times (with vanishing $\varepsilon_t$) forces any finitely competitive algorithm to move the center $x_t$ to $U_t$, thereby simulating the chasing problem. We state the full reduction for completeness.

For offline algorithms: Let $y_t$ be the location of an offline algorithm for chasing unions of two intervals. The CMC algorithm just follows it by setting $x_t = y_t$. The movement cost is the same, and the service cost is at most $5 \sum_t \varepsilon_t \le 5$. Thus, $\cost^*_{\text{CMC}} \le \cost^*_{\text{chasing}} + 5$, where $\cost^*_{\text{CMC}}$ and $\cost^*_{\text{chasing}}$ denote the optimal offline cost of the two instances, respectively.

Conversely, suppose we have an online CMC algorithm with center $x_t$ at time $t$ that pays $\cost_{\text{CMC}}$. Let $x'_t$ denote the closest point to $x_t$ in $U_t$. Since $x'_t \in U_t$, setting the location of the server to $x'_t$ is a feasible online algorithm for chasing unions of two intervals. By property 2 above, we have $\clust_t(x_t) \ge 2d(x_t, U_t) = 2d(x_t, x'_t)$. The movement can be bounded by
\begin{align*}
    \cost_{\text{chasing}}=\sum_t |x'_t-x'_{t-1}|
    &\le
    \sum_t |x_t-x_{t-1}|+2\sum_t d(x_t,x'_t)
    \\&\le
    \sum_t |x_t-x_{t-1}|+\sum_t \clust_t(x_t)
    \le
    \cost_{\text{CMC}}.
\end{align*}
Thus, any $C$-competitive algorithm for online CMC with non-monotone norms implies a $C$-competitive algorithm for chasing unions of two intervals, leading to a contradiction.
\end{proof}

\subsection{Non-Negative Norming Functionals}
Here, we show the existence of non-negative norming functionals for monotone norms.
\label{appendix: prelim}
\begin{proof}[Proof of \cref{lemma: positive dual}]
    Let \[\partial \|x \| = \{z : \|y\| \ge \|x\| + \langle z, y-x \rangle \text{ for all } y\}\] be the set of subgradients at $x$. Since $\|\cdot\|$ is monotone, there exists $z \in \partial \|x\|$ such that $z \ge 0$ for all $x \ge 0$.

    We will now show that $\partial \|x \| \subseteq \{z: \|z\|_* = 1, \langle z, x \rangle = \|x\|\}$.\footnote{In fact, the two sets are equal, but inclusion is enough for us.} Fix any $z \in \partial \|x \|$. We have
    \begin{align*}
    \|0\| \ge \|x\| + \langle z, 0-x \rangle &\implies \langle z, x \rangle \ge \|x\|, \\
    \|2x\| \ge \|x\| + \langle z, 2x-x \rangle &\implies \langle z, x \rangle \le \|x\|.
    \end{align*}
    Thus, $\langle z, x \rangle = \|x\|$. Now, for all $y$, we have
    \[
    \|y\| \ge \|x\| + \langle z, y-x \rangle = \|x\| + \langle z, y \rangle - \langle z, x \rangle = \langle z, y \rangle.
    \]
    Therefore,
    \[
    \|z\|_* = \max_{\|y\| \le 1} \langle z, y \rangle \le  \max_{\|y\| \le 1} \|y\| \le 1.
    \]

    Finally, since $z$ is a maximizer of $\langle z, x \rangle$, it must be that $\|z\|_* = 1$; otherwise, we can scale up $z$. This completes the proof that the set of norming functionals is a superset of subgradients, and we can choose a non-negative subgradient.
\end{proof}

\section{Reduction for Chasing Problems}
\label{appendix:mss}

Here we prove \cref{thm: mss}, the reduction for chasing problems. The algorithm is the same as \textsc{ShadowServe} for online CMC, except that we start from an $\alpha$-competitive $\mc F$-chasing algorithm. The pseudocode is in \cref{alg:shadowserv-mss}.

\begin{algorithm}[h]
\caption{\shadowserv for chasing problems}
\label{alg:shadowserv-mss}
\begin{algorithmic}[1]
\Require An $\mc F$-chasing algorithm $\mc A$ on $(X,d)$, initial location $x_0$
\State Let $y(\tau)$ denote the location of $\mc A$ at time $\tau \in \R_{\ge 0}$, initially $y(0) \gets x_0$
\For{rounds $t=1,2,\dots$}
  \State Observe request multiset $R_t \subseteq \mc F$ and the norm $\|\cdot\|$ on $\mathbb{R}^{R_t}$
  \ForAll{infinitesimal time intervals $[\tau,\tau + \d\tau) \subseteq (t-1,t]$}
    \State Let $d(\tau) \gets (d(y(\tau),F))_{F\in R_t}$
    \State Choose a norming functional $z(\tau)\in \mathbb{R}^{R_t}_{\ge 0}$ s.t. $\|z(\tau)\|_* = 1$ and $\langle z(\tau), d(\tau)\rangle = \|d(\tau)\|$
    \ForAll{requests $F \in R_t$ \textbf{independently}}
      \State Issue an MSS request $F$ to $\mc A$ with probability $z_F(\tau)\, \d\tau$
    \EndFor
  \EndFor
  \State Select $\tau_t \in \arg\min_{\tau\in(t-1,t]} \|d(\tau)\|$
  \State Set $x_t \gets y(\tau_t)$
\EndFor
\end{algorithmic}
\end{algorithm}

\begin{theorem}
    If $\mc A$ is a randomized $\alpha$-competitive algorithm for the $\mc F$-chasing problem against adaptive online adversaries, then $\shadowserv$ is a randomized $O(\alpha)$-competitive algorithm for the batched $\mc F$-chasing problem against adaptive online adversaries. If the norm is weighted $\ell_1$ at all times, then the reduction also works for oblivious adversaries.
\end{theorem}
\begin{proof}
Fix an adaptive online adversary for batched $\mc F$-chasing, and consider an adaptive online adversary for the (non-batched) $\mc F$-chasing problem that generates the request sequence as in our algorithm. Let $M^*$ denote the (random) cost of this adversary. Let $M_t^{\mc A}$ be the movement cost of $\mc A$ during round $t$ against this sequence, and let $M^\mc A$ be the overall movement cost of $\mc A$.

\begin{claim}
\label{claim:alg-mov-mss-mov}
The movement cost of \shadowserv is bounded by the movement cost of the $\mc F$-chasing algorithm: in every realization,
\[
\sum_t d(x_{t-1},x_t) \le \sum_t M_t^{\mc A}.
\]
\end{claim}

\begin{proof}
Fix a realization of the randomness. Let $\tau^{(1)}, \tau^{(2)}, \ldots$ be the sequence of times at which $\mc A$ moves, and let $\tau^{(0)}=0$. Then
\begin{align*}
\sum_t d(x_{t-1},x_t)
&= \sum_t d\bigl(y(\tau_{t-1}),y(\tau_t)\bigr) \\
&\le \sum_t \sum_{\tau^{(i)}\in (\tau_{t-1},\tau_t]} d\bigl(y(\tau^{(i-1)}),y(\tau^{(i)})\bigr) \\
&\le \sum_t M_t^{\mc A}. \qedhere
\end{align*}
\end{proof}

\begin{claim}
\label{claim:alg-serv-mss-mov}
The expected service cost is bounded by the expected movement cost of the $\mc F$-chasing algorithm: for all $t$, the corresponding optimal service time $\tau_t$ satisfies
\[
\E[\|d(\tau_t)\|] \le \E[M_t^{\mc A}].
\]
\end{claim}

\begin{proof}
The claim follows from the following inequalities:
\begin{align}
\E[\|d(\tau_t)\|]
&\le \E\left[\int_{t-1}^{t} \|d(\tau)\| \, \d\tau\right] \label{eq:mss-1} \\
&= \E\left[\int_{t-1}^{t} |\langle z(\tau), d(\tau)\rangle| \, \d\tau \right] \label{eq:mss-2} \\
&= \int_{t-1}^{t} \E\bigl[\langle z(\tau), (d(y(\tau),F))_{F\in R_t}\rangle\bigr] \, \d\tau \label{eq:mss-3} \\
&\le \E[M_t^{\mc A}] \label{eq:mss-4},
\end{align}
where inequality~\eqref{eq:mss-1} holds because we choose $\tau_t$ with the smallest service cost, equality~\eqref{eq:mss-2} follows from the definition of $z$, and equality~\eqref{eq:mss-3} follows from Fubini's theorem and the non-negativity of both $z(\tau)$ and $d(\tau)$. To see inequality~\eqref{eq:mss-4}, at time $\tau$, each request $F$ appears with rate $z_F(\tau)$, in which case $\mc A$ must incur a movement cost of at least $d(y(\tau),F)$ to serve it. Thus, in expectation, the instantaneous cost of $\mc A$ is at least $z_F(\tau)\, d(y(\tau),F)$ for each request $F$, and inequality (4) follows from integrating.
\end{proof}

\begin{claim}
\label{claim:mss-opt-bmss-opt}
The expected cost of the adaptive online adversary for $\mc F$-chasing is bounded by the cost of the adaptive online adversary for batched $\mc F$-chasing:
\[
\E[M^*] \le 2\E[\cost^*].
\]
\end{claim}

\begin{proof}
Let $x_t^*$ denote the (random) location of an adaptive online adversary for batched $\mc F$-chasing in round $t$, and let
\[
d_t^* := (d(x_t^*,F))_{F\in R_t}
\]
be the corresponding distance vector. Consider the following service strategy of the adaptive online adversary for $\mc F$-chasing: at round $t$, first move to $x_t^*$, and then whenever a request during round $t$ appears, move to the request set and then return to $x_t^*$. The expected cost of this strategy (conditioned on $d_t^*$) during round $t$ is
\begin{align*}
2 \int_{t-1}^{t} \E[\langle z(\tau), d_t^*\rangle \mid d_t^*] \, \d\tau
&\le 2 \int_{t-1}^{t} \left(\max_{\|z\|_*=1} |\langle z, d_t^* \rangle|\right)\,\d\tau \\
&= 2 \int_{t-1}^{t} \|d_t^*\| \, \d\tau
\end{align*}
where the inequality follows from $\|z(\tau)\|_*=1$. Further, this strategy moves between rounds exactly as $\OPT$ does. Thus, the total cost of this strategy is upper bounded by $2\cdot \cost^*$.
\end{proof}

Combining Claims \ref{claim:alg-mov-mss-mov}, \ref{claim:alg-serv-mss-mov}, and \ref{claim:mss-opt-bmss-opt} with the fact that $\mc A$ is $\alpha$-competitive against adaptive online adversaries, we obtain
\begin{align*}
    \cost=\sum_t \bigl(\E[d(x_{t-1},x_t)] + \E[\|d(\tau_t)\|]\bigr)
&\le \sum_t 2\E[M_t^{\mc A}]
\\&\le 2\alpha\cdot \E[M^*] + 2\beta
\le 4\alpha \cdot \E[\cost^*] + 2\beta.
\end{align*}

Finally, when the norms are weighted $\ell_1$ at every time, the norming functional $z$ is constant on each interval and does not depend on the current configuration of $\mc A$. Therefore, once the batched $\mc F$-chasing instance is fixed, the induced distribution over $\mc F$-chasing requests is fixed as well, and the reduction also applies when $\mc A$ is competitive only against oblivious adversaries.
\end{proof}

\subsection{Derandomizing Convex Body Chasing}
\label{appendix: derandomize}
Here we prove \cref{cor:convex body} and give a deterministic $O(d)$-competitive algorithm for batched convex body chasing. To do so, we show the following derandomization lemma.

\begin{lemma}
    If there is an $\alpha$-competitive randomized algorithm for batched convex body chasing against an oblivious adversary, then there is an $\alpha$-competitive deterministic algorithm for batched convex body chasing.
\end{lemma}

\begin{proof}
Let $X_t \in \R^d$ be the random position of a randomized algorithm in round $t$, and define the deterministic position $y_t := \E[X_t]$. We show that the deterministic algorithm that outputs $y_t$ in each round has cost no larger than the expected cost of the randomized one.

For the movement cost, by convexity of the norm,
\[
\|y_t-y_{t-1}\|
= \|\E[X_t-X_{t-1}]\|
\le \E[\|X_t-X_{t-1}\|].
\]

Now fix a convex body $F \subseteq \R^d$. The function $x \mapsto d(x,F)$ is convex, so by Jensen's inequality,
\[
d(y_t,F)=d(\E[X_t],F)\le \E[d(X_t,F)].
\]
Hence, for the distance vector to the batch $R_t$, we have
\[
(d(y_t,F))_{F\in R_t} \le \E[(d(X_t,F))_{F\in R_t}]
\]
coordinatewise. Since the batch norm is monotone and convex,
\[
\|(d(y_t,F))_{F\in R_t}\|
\le \left\|\E[(d(X_t,F))_{F\in R_t}]\right\|
\le \E\!\left[\|(d(X_t,F))_{F\in R_t}\|\right].
\]

Therefore, in every round $t$,
\[
\|y_t-y_{t-1}\| + \|(d(y_t,F))_{F\in R_t}\|
\le
\E\!\left[\|X_t-X_{t-1}\| + \|(d(X_t,F))_{F\in R_t}\|\right].
\]
Summing over all rounds proves that the deterministic algorithm $(y_t)_t$ has total cost at most the expected cost of the randomized algorithm.
\end{proof}

\cref{cor:convex body} then follows from applying \cref{thm: mss} to the convex body chasing problem and using the lemma.

\section{Discretized Implementation of Line Algorithms}
\label{appendix:discretize}

We describe a discrete implementation of (the more complex) \radAggrDC{} with only a
constant-factor loss in the competitive guarantee, and then give the
analogous implementation of \AggrDC. At a very high level, we split the interval into small enough segments and run the algorithm at discrete time steps. The segments are so short relative to the current norm that the error due to discretization can be absorbed.

\paragraph{The Algorithm.}
Initialize the internal configuration at $C_0$ and retain the
left-to-right indices of the centers. Break nearest-center ties in
favor of the smallest index. Use fixed deterministic rules for norming
functionals and all other ties.
For a configuration $S$ in round $t$, write $\rad(S)$ for its
nearest-center radius vector, with radius zero for an empty cluster.

Fix a nonempty round $t$. We split the interval $(t-1, t]$ into $N_t$ segments of length $\varepsilon_t$, where
\[
N_t:=\lceil 8k\|\mathbf 1\|\rceil,\qquad
\varepsilon_t:=1/N_t.
\]
Set $S^0$ to the current internal configuration. Let $S^\ell$ be the configuration at the beginning of
iteration $\ell$, for $\ell=0,\ldots,N_t-1$.
By \Cref{lemma: positive dual}, choose a nonnegative norming functional
$z^\ell$ for $\rad(S^\ell)$ satisfying
\[
\|z^\ell\|_*=1,\qquad
\langle z^\ell,\rad(S^\ell)\rangle=\|\rad(S^\ell)\|.
\]
For each nonempty cluster $c$ of $S^\ell$, choose a farthest request
$r_c^\ell$. Freeze these requests and $z^\ell$ during the iteration,
and process the nonempty starting clusters in center order.

For each $c$, run single-request Double Coverage toward $r_c^\ell$ for
$\varepsilon_t z_c^\ell$ units of time, with speed equal to the current
nearest-center distance to $r_c^\ell$. Select the rightmost indexed center
strictly to its left and the leftmost indexed center strictly to its
right, when they exist. The distance decays exponentially, so this is
an explicit update: each selected center moves by
$1-e^{-\varepsilon_t z_c^\ell}$ times the distance from $r_c^\ell$ to the
configuration before its motion. A request already occupied induces
no movement. Since $1-e^{-\varepsilon_t z^\ell_c} \le \varepsilon_tz^\ell_c \le \varepsilon_t \|\mathbf 1\| < 1$, no center passes the request, and the center order is
preserved.
Let $S^{\ell+1}$ be the configuration after iteration $\ell$.

At the end of the round, output the configuration with the smallest clustering cost, $C_t=S^{\ell^*}$, where $\ell^*\in\arg\min_{0\le \ell<N_t}\|\rad(S^\ell)\|$. The internal configuration is not rewound: the next round starts from $S^{N_t}$. We now show that the asymptotic competitive ratio is the same. The constant $32$ is not optimized.

\begin{theorem}
\label{thm:discrete-radaggrdc}
The discrete algorithm is $32k$-competitive.
\end{theorem}

\begin{proof}
Fix an offline solution. We use the potential from the proof of
\Cref{thm:radaggrdc},
\[
\Phi(S,C^*):=2\bigl(k\Match(S,C^*)+\Sigma(S)\bigr).
\]
As there, offline movement by a distance $L$ increases $\Phi$ by at most
$2kL$.

Fix a round $t$ and suppress only its index, writing
$\varepsilon:=\varepsilon_t$ and $N:=N_t$.
Write $C^*=C_t^*$ and let $\rad^*$ be its radius vector, using its
chosen assignment of requests. For each iteration $\ell$ and processed
center $c$, let $D_c^\ell$ be the integral of the nearest-center distance
to its frozen request $r_c^\ell$ during its motion, $M_c^\ell$ the movement
incurred, and $\Delta_c^\ell\Phi$ the change in potential. Write
$D^\ell:=\sum_c D_c^\ell$ and $M^\ell:=\sum_c M_c^\ell$.
The distance integral $D_c^\ell$ equals the distance traveled by each
selected center, so $M^\ell\le2D^\ell$.

\paragraph{Single-Request Charges.}
Integrating \Cref{claim:aggrdc-local} directly over the motion gives
\begin{equation*}
\Delta_c^\ell\Phi+M_c^\ell+D_c^\ell
\le4k\varepsilon z_c^\ell\,d(r_c^\ell,C^*).
\label{eq:discrete-rad-local}
\end{equation*}
Also, the one- and two-center calculations in its proof give
$\Delta_c^\ell\Phi+M_c^\ell+D_c^\ell\le4kD_c^\ell$.

These single-request bounds incur no discretization loss.
The discrepancy is that $z^\ell$ is a norming functional for
$\rad(S^\ell)$ and the witnesses $r_c^\ell$ are chosen at $S^\ell$, whereas
$D_c^\ell$ uses distances to configurations reached during iteration $\ell$.
For each frozen request $r_c^\ell$, these distances differ from
$d(r_c^\ell,S^\ell)=\rad_c(S^\ell)$ by at most $M^\ell$.
This movement error enters both the aggregation of charges and the
comparison with the clustering cost at $S^\ell$.

\paragraph{Aggregating the Charges.}
For each nonempty cluster of $S^\ell$, let $\chi^\ell(c)$ be the offline
center assigned to $r_c^\ell$, extending $\chi^\ell$ arbitrarily on empty
clusters. Define the amortized charge per unit time by
\[
s_c^\ell:=\frac{\max\{0,\Delta_c^\ell\Phi+M_c^\ell+D_c^\ell\}}{\varepsilon z_c^\ell}
\]
for $z_c^\ell>0$, setting all charges and costs to zero for skipped
clusters or zero-duration motions. The preceding distance comparison
gives $D_c^\ell\le\varepsilon z_c^\ell(\rad_c(S^\ell)+M^\ell)$, and hence
\[
s_c^\ell\le4k(\rad_c(S^\ell)+M^\ell),\qquad
s_c^\ell\le4k\rad^*_{\chi^\ell(c)}.
\]
To apply the starting configuration's interval geometry, remove the
movement error by putting $b_c^\ell:=\max\{0,s_c^\ell-4kM^\ell\}$. Then
\[
b_c^\ell\le4k\rad_c(S^\ell),\qquad b_c^\ell\le4k\rad^*_{\chi^\ell(c)}.
\]

The geometric estimate in the proof of \cref{claim:radaggrdc-local},
applied at $S^\ell$, gives
$\sum_{c\in(\chi^\ell)^{-1}(j)}b_c^\ell\le16k\rad_j^*$ for each offline
center $j$. By \Cref{lem:aggregation} and monotonicity,
$\|b^\ell\|\le16k\|\rad^*\|$.
Since $s^\ell\le b^\ell+4kM^\ell\mathbf 1_k$, it follows that
\[
\|s^\ell\|\le16k\|\rad^*\|+4kM^\ell\|\mathbf 1_k\|.
\]
Let $\Delta^\ell\Phi:=\Phi(S^{\ell+1},C^*)-\Phi(S^\ell,C^*)$.
Summing the submoves and using $\|z^\ell\|_*=1$ gives
\begin{align*}
\Delta^\ell\Phi+M^\ell+D^\ell
&\le\varepsilon\langle z^\ell,s^\ell\rangle\le16k\varepsilon\|\rad^*\|
   +4k\varepsilon M^\ell\|\mathbf 1_k\|.
\end{align*}
Since $\varepsilon\le1/(8k\|\mathbf 1\|)$, this error term is at most
$M^\ell/2$, so
\begin{equation}
\Delta^\ell\Phi+\tfrac12M^\ell+D^\ell
\le16k\varepsilon\|\rad^*\|.
\label{eq:discrete-rad-iteration}
\end{equation}

\paragraph{Clustering and Round Cost.}
The norming identity for $z^\ell$ is exact at $S^\ell$.
Replacing the starting radii by the distance
integrals $D_c^\ell$ incurs some error, which is bounded by $M^\ell$. Integrating the same fixed-request distance
comparison over each motion, and using $M^\ell\le2D^\ell$, gives
\[
\varepsilon\|\rad(S^\ell)\|
=\varepsilon\langle z^\ell,\rad(S^\ell)\rangle
\le D^\ell+\varepsilon M^\ell\|\mathbf 1_k\|
\le(1+2\varepsilon\|\mathbf 1_k\|)D^\ell
\le2D^\ell.
\]
Thus this comparison loses at most a factor of two. Setting
$\Psi:=2\Phi$ also restores the full movement coefficient after the
earlier error absorption. Combining with
\eqref{eq:discrete-rad-iteration} gives
\[
\Delta^\ell\Psi+M^\ell+\varepsilon\|\rad(S^\ell)\|
\le32k\varepsilon\|\rad^*\|.
\]
Let $M_t:=\sum_{\ell=0}^{N-1}M^\ell$ be the total internal movement in round $t$ and
$\Delta_t\Psi$ its potential change with $C_t^*$ held fixed.
Since the algorithm selects the best iteration-start configuration,
\[
\|\rad(C_t)\|
\le\varepsilon\sum_{\ell=0}^{N-1}\|\rad(S^\ell)\|.
\]
Summing the iteration inequalities therefore gives
\[
\Delta_t\Psi+M_t+\|\rad(C_t)\|
\le32k\|\rad^*\|.
\]
The configurations output in successive rounds occur in chronological
order along the internal movement path, so their total movement is
at most $\sum_t M_t$.
Finally, offline movement increases $\Psi$ by at most $4k$ times its
movement cost. Summing over rounds, using $\Psi\ge0$, proves
\[
\cost\le32k\,\clust^*+4k\,\mov^*+\Psi_0
\le32k\,\cost^*+\Psi_0,
\]
where $\Psi_0$ is the potential of the fixed initial configurations,
independent of the request sequence.
\end{proof}

For \AggrDC, process all requests in $R_t$ in any fixed order instead
of one farthest request per cluster, and choose $z^\ell$ as a norming
functional for the starting distance vector. Take
$N_t=\lceil2\|\mathbf 1_{R_t}\|\rceil$ and
$\varepsilon_t=1/N_t$. With $M^\ell$ and $D^\ell$ defined analogously,
integrating \Cref{claim:aggrdc-local} and summing directly gives
\[
\Delta^\ell\Phi+M^\ell+D^\ell
\le4k\varepsilon_t\langle z^\ell,d(C_t^*)\rangle
\le4k\varepsilon_t\|d(C_t^*)\|.
\]
The same clustering comparison gives
$\varepsilon_t\|d(S^\ell)\|\le2D^\ell$.
Doubling the potential and selecting the best iteration-start
configuration therefore proves that this implementation is
$8k$-competitive.

\section{Associations for Norm-of-Radii Objective}
\label{appendix: association}
Recall that in the norm-of-radii objective, the clustering cost is determined not only by the location of the centers, but also by the associations between requests and the centers. We first show that on the line, greedy assignment is optimal up to a constant factor.

\begin{theorem}\label{thm:nn-constant-factor-line}
Let $c_1<\cdots<c_k$ be fixed centers on the line and let $R$ be a multiset of request points.
For an assignment $\pi:R\to[k]$, let $\rad(\pi)\in\R_{\ge 0}^k$ be the radius vector, where
\[
\rad_i(\pi):=\max\{\,|x-c_i|:\ x\in R,\ \pi(x)=i\,\},
\]
with $\rad_i(\pi)=0$ if no request is assigned to $c_i$.
Let $\|\cdot\|$ be any monotone symmetric norm on $\R^k$.
Let $\pi^{\mathrm{nn}}$ be the nearest-center assignment (ties broken arbitrarily).
Then $\|\rad(\pi^{\mathrm{nn}})\|\le O(1)\cdot \min_{\pi}\|\rad(\pi)\|$.
\end{theorem}

\begin{proof}
It helps to think that the left and right radius vectors are paid separately. For an assignment $\pi$, define the \emph{left} and \emph{right} radii vectors $a(\pi),b(\pi)\in\R_{\ge 0}^k$ by
\begin{align*}
a_i(\pi)&:=\max\{\,c_i-x:\ x\in R,\ \pi(x)=i,\ x\le c_i\,\},\\
b_i(\pi)&:=\max\{\,x-c_i:\ x\in R,\ \pi(x)=i,\ x\ge c_i\,\},
\end{align*}
with the convention that the max over an empty set is $0$.
Then $\rad_i(\pi)=\max\{a_i(\pi),b_i(\pi)\}$ and $a_i(\pi)+b_i(\pi)$ is the diameter of the assigned interval around $c_i$.

Coordinatewise we have
\[
\frac{a_i(\pi)+b_i(\pi)}{2}\ \le\ \max\{a_i(\pi),b_i(\pi)\}\ \le\ a_i(\pi)+b_i(\pi).
\]
By monotonicity of $\|\cdot\|$, this implies
\begin{equation*}\label{eq:rad-vs-sum}
\frac12\|a(\pi)+b(\pi)\|\ \le\ \|\rad(\pi)\|\ \le\ \|a(\pi)+b(\pi)\|.
\end{equation*}
Moreover, by the triangle inequality and monotonicity,
\begin{equation*}\label{eq:sum-vs-split}
\|a(\pi)+b(\pi)\|\ \le\ \|a(\pi)\|+\|b(\pi)\|
\ \le\ 2\,\|a(\pi)+b(\pi)\|.
\end{equation*}
Thus, up to a constant factor, it suffices to analyze the objective $\|a(\pi)\|+\|b(\pi)\|$. Let $\pi^*$ minimize $\|a(\pi)\|+\|b(\pi)\|$, and write $a^*:=a(\pi^*)$, $b^*:=b(\pi^*)$.

\begin{claim}
   In $\pi^*$, no cluster contains another center. Formally, for every $i \ne j$, \[c_j \notin [c_i - a_i(\pi^*), c_i + b_i(\pi^*)].\]
\end{claim}
\begin{proof}
    Suppose for the purpose of contradiction that some center $c_i$ includes another center $c_{j}$ to its right. It should be easy to see that $b_j = 0$.

    Let $x$ be the rightmost request currently assigned to $c_i$ that lies to the right of $c_j$. Consider an alternative assignment which reassigns all requests in $(c_{j},x]$ from $c_i$ to $c_{j}$, which will decrease $b_i$ and increase $b_{j}$. Let $\Delta$ be the amount by which $b_i$ decreases; then $b_{j}$ increases by at most $\Delta$. Therefore the vector $b$ changes only in two coordinates, from $(b_i,0)$ to $(b_i-\Delta, \Delta)$. By convexity and symmetry of the norm (cf. \cref{lem:aggregation}), moving mass from a coordinate to a $0$-valued one cannot increase the norm; hence $\|b\|$ does not increase, while $\|a\|$ is unchanged. This contradicts optimality unless $\Delta=0$, i.e., the cluster of $c_i$ does not extend past $c_{j}$. Thus, we may assume $\pi^*$ is uncrossed and no cluster contains another center.
\end{proof}

Under an uncrossed, no-inclusion assignment, there is a boundary point $p_i$ between each consecutive pair of centers $c_i,c_{i+1}$ such that
all requests $\le p_i$ (in that local region) are assigned to $c_i$ and all requests $\ge p_i$ are assigned to $c_{i+1}$.
Let $m_i:=\frac{c_i+c_{i+1}}{2}$ be the midpoint (the nearest-neighbor boundary).

We now modify $\pi^*$ by shifting each boundary $p_i$ to $m_i$, obtaining the nearest-neighbor assignment $\pi^{\mathrm{nn}}$.
Consider a fixed $i$. Shifting the boundary from $p_i$ to $m_i$ only moves requests between $c_i$ and $c_{i+1}$; hence it can only affect
$b_i$ and $a_{i+1}$.
If $p_i\le m_i$, then $c_i$ gains some requests in $[p_i,m_i]$, so only $b_i$ may increase, and the increase is at most $m_i-p_i$.
If $p_i\ge m_i$, then $c_{i+1}$ gains some requests in $[m_i,p_i]$, so only $a_{i+1}$ may increase, and the increase is at most $p_i-m_i$.
In both cases,
\[
|p_i-m_i|\ \le\ \max\{\,b_i^*,\ a_{i+1}^*\,\},
\]
because $p_i$ lies between the rightmost request served by $c_i$ and the leftmost request served by $c_{i+1}$, and the no-inclusion property ensures these extremal requests sit on the correct sides of the centers.

Therefore, after shifting all boundaries, we obtain vectors $a^{\mathrm{nn}}:=a(\pi^{\mathrm{nn}})$ and $b^{\mathrm{nn}}:=b(\pi^{\mathrm{nn}})$ satisfying the coordinatewise bounds
\[
b_i^{\mathrm{nn}} \le b_i^* + \max\{b_i^*, a_{i+1}^*\}\ \le\ 2b_i^* + a_{i+1}^*,
\qquad
a_{i+1}^{\mathrm{nn}} \le a_{i+1}^* + \max\{b_i^*, a_{i+1}^*\}\ \le\ b_i^* + 2a_{i+1}^*.
\]
By monotonicity and the triangle inequality of $\|\cdot\|$,
\begin{equation*}\label{eq:one-sided-constant}
\|a^{\mathrm{nn}}\|+\|b^{\mathrm{nn}}\|\ \le\ 3\bigl(\|a^*\|+\|b^*\|\bigr).
\end{equation*}

Thus, we overall obtain
\[
\|\rad(\pi^{\mathrm{nn}})\|
\ \le\ \|a^{\mathrm{nn}}+b^{\mathrm{nn}}\|
\ \le\ \|a^{\mathrm{nn}}\|+\|b^{\mathrm{nn}}\|
\ \le\ 3\bigl(\|a^*\|+\|b^*\|\bigr)
\ \le\ 6\,\|a^*+b^*\|
\ \le\ 12\,\|\rad(\pi^*)\|.
\]
\end{proof}

We now show that on general metrics, greedy assignment could cost $\Theta(k)$ times as much as the optimal assignment.

\begin{theorem}
    On general metrics, there exist instances and a norm such that $\|\rad(\pi^{\mathrm{nn}})\| = \Omega(k)\cdot \min_{\pi}\|\rad(\pi)\|$. For all instances, $\|\rad(\pi^{\mathrm{nn}})\| = O(k)\cdot \min_{\pi}\|\rad(\pi)\|$.
\end{theorem}
\begin{proof}
    We first show an $O(k)$ upper bound. Let $D = \max_{r \in R} d(r, C)$. Then the radii of the greedy clusters do not exceed $D$, and there is at least one radius in the optimal clustering that is at least $D$. Hence,
    \[
    \|\rad(\pi^{\mathrm{nn}})\| \le D \|1\| \le k D \|(1, 0, \ldots, 0)\| \le k\|\rad(\pi^*)\|.
    \]

We now give an example reaching the upper bound in \cref{fig:star-lower-bound}. Consider the shortest-path metric of a star with $k-1$ branches. Each leaf is at a distance of $2$ from the hub.
Place $k-1$ centers $c_1,\dots,c_{k-1}$ on the leaves and $c_k$ on the hub. For each $i\in[k-1]$, place one request $r_i$ on branch $i$ at a distance of $1+\varepsilon$ from the hub $c_k$ (so $d(r_i,c_k)=1+\varepsilon$)
and hence at a distance of $1-\varepsilon$ from the leaf center $c_i$ (so $d(r_i,c_i)=1-\varepsilon$).
Then each $r_i$ is strictly closer to $c_i$ than to $c_k$. Therefore, the greedy assignment will assign $\pi(r_i) = c_i$, resulting in cost $(k-1)(1-\varepsilon)$. The optimal assignment, however, clusters all requests with $c_k$ with cost $1+\varepsilon$, and thus the ratio tends to $k-1$.
\end{proof}
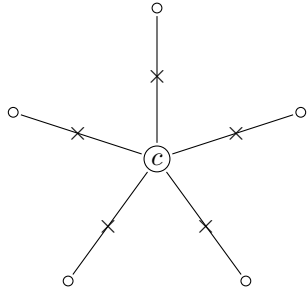
\begin{figure}[t]
\centering
\begin{tikzpicture}[scale=1.0, every node/.style={font=\small}]
  \node[draw, circle, inner sep=1.3pt] (hub) at (0,0) {$c$};

  \def\L{2}      %
  \def\mid{1.1}    %
  \def\sh{1pt}   %

  \foreach \i/\ang in {1/90,2/18,3/-54,4/-126,5/162} {
    \coordinate (leaf\i) at ({\L*cos(\ang)},{\L*sin(\ang)});
    \coordinate (req\i)  at ({\mid*cos(\ang)},{\mid*sin(\ang)});

    \node[draw, circle, inner sep=1.3pt] (C\i) at (leaf\i) {};

    \node (R\i) at (req\i) {$\times$};

    \draw[shorten >=\sh, shorten <=\sh] (hub) -- (C\i);
  }

\end{tikzpicture}
\caption{Star-metric lower-bound instance. Circles are centers, and crosses are requests. $d(c,c_i)=2,\quad d(r_i,c)=1+\varepsilon,\quad d(r_i,c_i)=1-\varepsilon$.}
\label{fig:star-lower-bound}
\end{figure}

\section{(A)symmetric Norm-of-Radii Objective}
\label{appendix: asymm}
In the analysis of our algorithm \radAggrDC, we needed symmetry to permute the coordinates of the cost vector. We now explain that, in order to consider asymmetric norms, something fundamentally different from the Double Coverage algorithm is required. In particular, algorithms like Double Coverage will never swap the relative order of two centers, which may be problematic in the weighted case.

Consider the following example with 2 centers and the norm being $B\cdot \rad_1 + \rad_2$ for some large $B$. The expensive center $c_1$ starts on the left. Consider three requests with coordinates $0, 1, B$. The requests are repeated so many times that we only care about the clustering configuration that we converge to. An optimal offline algorithm will switch the locations of the two centers to make the cheap $c_2$ cluster $\{0, 1\}$ and the expensive $c_1$ cluster $B$ alone, yielding a total clustering cost of $1/2$. If we do not switch (as in the Double Coverage algorithm) and keep the expensive center on the left, then we pay at least $\Theta(B)$: either the expensive center clusters $\{0, 1\}$, or the cheap center clusters $\{1, B\}$, which could be arbitrarily more expensive and result in an unbounded competitive ratio.

\end{document}